\documentclass[11pt]{article}

\RequirePackage{etex}

\usepackage[margin=1in]{geometry}   
\usepackage{amsfonts}
\usepackage{amsmath} 
\usepackage{amssymb}
\usepackage{centernot}
\usepackage{booktabs}
\usepackage{multirow}
\usepackage{algorithm}% http://ctan.org/pkg/algorithms
\usepackage{algorithmicx}% http://ctan.org/pkg/algorithmic
\usepackage[flushleft]{threeparttable}

\usepackage[onehalfspacing]{setspace} 
\usepackage{float}
\usepackage{caption}
\usepackage{graphicx}
\usepackage{adjustbox}
\usepackage{lscape}

\usepackage{adjustbox}

\usepackage{xcolor}

\usepackage{relsize}

\usepackage{bbm}
\usepackage{subfig}

\definecolor{forestgreen}{RGB}{34,139,34}
\usepackage{hyperref}
\hypersetup{
    colorlinks=true,
    linkcolor=magenta,
    filecolor=magenta, 
    citecolor=forestgreen,      
    urlcolor=blue
}

\usepackage{authblk}

\usepackage{amsthm}
\newtheorem{theorem}{Theorem}

\usepackage{xpatch}
\makeatletter
\xpatchcmd{\proof}{\@addpunct{.}}{\@addpunct{:}}{}{}
\makeatother

\makeatletter
\def\@hangfrom#1{\setbox\@tempboxa\hbox{{#1}}%
      \hangindent 0pt%\wd\@tempboxa
      \noindent\box\@tempboxa}
\makeatother

\makeatletter
\newcommand{\vast}{\bBigg@{3}}
\newcommand{\Vast}{\bBigg@{4}}
\makeatother

\makeatletter
\newcommand*{\indep}{%
  \mathbin{%
    \mathpalette{\@indep}{}%
  }%
}
\newcommand*{\nindep}{%
  \mathbin{%                   % The final symbol is a binary math operator
    \mathpalette{\@indep}{\not}% \mathpalette helps for the adaptation
  }%
}
\newcommand*{\@indep}[2]{%
  \sbox0{$#1\perp\m@th$}%        box 0 contains \perp symbol
  \sbox2{$#1=$}%                 box 2 for the height of =
  \sbox4{$#1\vcenter{}$}%        box 4 for the height of the math axis
  \rlap{\copy0}%                 first \perp
  \dimen@=\dimexpr\ht2-\ht4-.2pt\relax
  \kern\dimen@
  {#2}%
  \kern\dimen@
  \copy0 %                       second \perp
} 
\makeatother

\DeclareMathOperator{\E}{\textnormal{\mbox{E}}}

\usepackage[utf8]{inputenc}
\DeclareUnicodeCharacter{200E}{}

\usepackage{tikz}
\usetikzlibrary{positioning,shapes.geometric}

\usepackage{setspace}
\usepackage{cite}

\usepackage[stable]{footmisc}

\usepackage{rotating}

\makeatletter
\def\@seccntformat#1{\@ifundefined{#1@cntformat}%
   {\csname the#1\endcsname\quad}  % default
   {\csname #1@cntformat\endcsname}% enable individual control
}
\let\oldappendix\appendix %% save current definition of \appendix
\renewcommand\appendix{%
    \oldappendix
    \newcommand{\section@cntformat}{\appendixname~\thesection\quad}
}
\makeatother

\usepackage[absolute,showboxes]{textpos}

\TPMargin{5pt}

\usepackage{thmtools, thm-restate}
\usepackage{datetime}

\usepackage{sectsty}
\subsectionfont{\noindent\textbf\itshape}
\subsubsectionfont{\normalfont\itshape}

\begin{document}

\title{Sensitivity analysis for studies transporting prediction models}

\author[1]{Jon A. Steingrimsson}
\author[2,3]{Sarah E. Robertson}
\author[2-4]{Issa J. Dahabreh}

\affil[1]{Department of Biostatistics, Brown University School of Public Health, Providence, RI }
\affil[2]{CAUSALab, Harvard T.H. Chan School of Public Health, Boston, MA}
\affil[3]{Department of Epidemiology, Harvard T.H. Chan School of Public Health, Boston, MA}
\affil[4]{Department of Biostatistics, Harvard T.H. Chan School of Public Health, Boston, MA}

\maketitle{}

\clearpage

\vspace*{1in}

\begin{abstract}
\noindent
\linespread{1.3}\selectfont
We consider the estimation of measures of model performance in a target population when covariate and outcome data are available on a sample from some source population and covariate data, but not outcome data, are available on a simple random sample from the target population. When outcome data are not available from the target population, identification of measures of model performance is possible under an untestable assumption that the outcome and population (source or target population) are independent conditional on covariates. In practice, this assumption is uncertain and, in some cases, controversial. Therefore, sensitivity analysis may be useful for examining the impact of assumption violations on inferences about model performance. Here, we propose an exponential tilt sensitivity analysis model and develop statistical methods to determine how sensitive measures of model performance are to violations of the assumption of conditional independence between outcome and population. We provide identification results and estimators for the risk in the target population, examine the large-sample properties of the estimators, and apply the estimators to data on individuals with stable ischemic heart disease.
\end{abstract}

\clearpage

%%%%%%%%%%%%%%%%%%%%%%%%%%%%%%%%%%%%%%%%%%%%%%%%%%%%%%%%%%%%%%%%%%%%%%%%%%%%%%
\section{Introduction}
%%%%%%%%%%%%%%%%%%%%%%%%%%%%%%%%%%%%%%%%%%%%%%%%%%%%%%%%%%%%%%%%%%%%%%%%%%%%%%

Users of prediction models are typically interested in obtaining model-derived predictions in a target population of substantive interest. However, the data used for model building and evaluation of model performance (i.e., the source data) are often not a random sample from the target population (e.g.,~due to convenience sampling or the two data sources coming from different geographic regions or healthcare systems). When prediction error modifiers \cite{steingrimsson2023transporting}, that is, variables that affect model performance, have a different distribution between the source population and the target population, measures of model performance calculated using data from the source population are not representative of model performance in the target population. 

It is possible to estimate model performance in the target population under the untestable assumption that the outcome is independent of the population (source or target) given the observed covariates \cite{shimodaira2000improving, sugiyama2007covariate, sugiyama2012machine, steingrimsson2023transporting} using the source data and covariate data from the target population, even when outcome information is unavailable in the target population. The assumption that the outcome is independent of the population, however, is untestable using the observed data and will often be uncertain, or even controversial, in practical applications. Therefore, it is useful to conduct sensitivity analyses to determine how sensitive conclusions are to violations of the conditional transportability condition.

There is a large literature on sensitivity analysis for missing data \cite{robins2000c,rotnitzky1998semiparametric,scharfstein2018globalBiometrics,scharfstein2018globalSMMR,scharfstein2021global,daniels2008missing} and unmeasured confounding in observational studies\cite{cornfield1959smoking,brumback2004sensitivity,klungsoyr2009,rosenbaum1983assessing,scharfstein2021semiparametric}. In addition, a smaller but growing literature consider methods for sensitivity analyses when extending (i.e., generalizing or transporting) inferences about treatment effects from a randomized trial to a target population \cite{nguyen2017sensitivity, nguyen2018sensitivity,dahabreh2019sensitivitybiascor, dahabreh2022global, duong2023sensitivity}. To our knowledge there is no prior work developing sensitivity analysis methods for evaluating the performance of prediction models in a target population. This task involves different target parameters and requires different identifiability results and estimation procedures than transportability of measures of model performance.

Here, we develop global sensitivity analysis methods for loss-based measures of model performance in the target population using an exponential tilt model \cite{scharfstein2018globalBiometrics, scharfstein2018globalSMMR, scharfstein2021global, scharfstein2021semiparametric}. Global sensitivity analysis allows for evaluation of how big the violation of a core assumption needs to be in order for conclusions to change \cite{scharfstein2014global}. We provide identification results and derive estimators and large sample properties of the estimators for both ``nested'' and ``non-nested'' sampling designs \cite{dahabreh2021study}. We show how the range of the sensitivity parameter can be selected by hypothesizing about a reasonable range of prevalence rate of the outcome in the target population. We illustrate the methods using data on individuals with stable ischemic heart disease.

%%%%%%%%%%%%%%%%%%%%%%%%%%%%%%%%%%%%%%%%%%%%%%%%%%%%%%%%%%%%%%%%%%%%%%%%%%%%%%
\section{Goals of the analysis, study design, and data structures}
%%%%%%%%%%%%%%%%%%%%%%%%%%%%%%%%%%%%%%%%%%%%%%%%%%%%%%%%%%%%%%%%%%%%%%%%%%%%%%

Let $Y$ be a univariate outcome assessed at the end of the study (e.g.,~binary, count, or continuous) and $X \in \mathcal{X}$ a baseline covariate vector. Under a non-nested sampling design \cite{dahabreh2020studydesigns}, we assume that we have access to a random sample of outcome and covariate information from the source population $\{(X_i, Y_i), i = 1, \ldots, n_1\}$ and a separately obtained random sample of covariates, but no outcome information, from the target population $\{X_i, i = 1, \ldots, n_0\}$. This setup does not restrict the data from the source population to be obtained from a formal sampling process and can be thought of as if sampled from some underlying hypothetical super-population that is potentially not well characterized (e.g.,~as is the cases when using convenience sampling) \cite{robins1988confidence, dahabreh2019efficient}. However, we assume that the target population data is representative of a target population of substantive interest. Let $S$ be an indicator whether the data is from the source population ($S=1$ if from the source population and $S =0$ if from the target population). Under this setup, the combined data from the source and target population is
\[
\mathcal{O} = \{X_i, S_i, S_i \times Y_i, i = 1, \ldots, n = n_1 + n_0\}.
\]
Let $X^*$ be a subset of $X$ that is used for constructing a prediction model and let $h(X^*, \beta)$ be a prediction model for the conditional expectation $\E[Y|X^*,S=1]$ indexed through the unknown parameter $\beta \in \mathcal{B}$. We use $\widehat \beta$ to denote an estimator for $\beta$ and $h(X^*, \widehat \beta)$ as the fitted model. Throughout this paper, we do \emph{not} assume that the model $h(X^*, \beta)$ is correctly specified; thus, our results also hold for misspecified models (under the assumptions listed in Sections \ref{sec:sensitivity_analysis} and  \ref{sec:id-sens}).  We assume that the model is built (i.e.,~$\beta$ is estimated) on a dataset that is independent of the data used to evaluate the model and we use $f(\cdot)$ to generically denote densities. % This includes the setting where the model is built using a training set that is independent of the test set used for estimation of model performance and where we wish to evaluate the performance of an external model (e.g.,~the Gail model \cite{gail1989projecting} or the Framingham risk score \cite{lloyd2004framingham}). 

We focus on estimation of loss-based measures of model performance in the target population. A loss function $L(Y, h(X^*, \widehat \beta))$ quantifies the discrepancy between the observed outcome $Y$ and model-derived predictions $h(X^*, \widehat \beta)$. Common examples include the mean squared error, absolute deviation, and Brier loss functions \cite{brier1950verification}. Our target parameter is the expected loss (risk) in the target population; in the non-nested design that parameter is $\E[L(Y,h(X^*, \widehat \beta))|S=0]$.

An alternative approach is to use a nested design \cite{dahabreh2020studydesigns, steingrimsson2023transporting} where the source population is nested within a cohort that is a sample from the target population (e.g.,~using record linkage of the source data with data from the target population). For nested designs, we assume that covariate information is available from the entire cohort but outcome information is only available from the source population data. For nested designs \cite{dahabreh2021study}, the target parameter is $\E[L(Y, h(X^*, \widehat \beta))]$. Sampling designs for both nested and non-nested designs have been discussed elsewhere \cite{dahabreh2020studydesigns, steingrimsson2023transporting} and all expectations and probabilities are under the distribution induced by the sampling design.

%%%%%%%%%%%%%%%%%%%%%%%%%%%%%%%%%%%%%%%%%%%%%%%%%%%%%%%%%%%%%%%%%%%%%%%%%%%%%
\section{Identification}
\label{sec:id}
%%%%%%%%%%%%%%%%%%%%%%%%%%%%%%%%%%%%%%%%%%%%%%%%%%%%%%%%%%%%%%%%%%%%%%%%%%%%%%

\subsection{Identifiability conditions}

The following two conditions are sufficient for identifiability of loss-based measures of model performance using the observable data $\mathcal{O}$ for non-nested designs \cite{steingrimsson2023transporting}.
\begin{itemize}
    \item[A1.] Positivity: $\Pr[S=1|X=x] >0$ for all $x \in \mathcal{X}$ that have a positive density in the target population $f_{X,S}(x, S=0) > 0$. Condition A1 is in principle testable using the observed data, but performing tests for the validity of that assumption can be challenging with high-dimensional covariates \cite{petersen2012diagnosing}.
    \item[A2.] Conditional transportability: $Y \indep S|X$. This key condition is untestable using the observed data because outcome information from the target population is unavailable; therefore, in many applications, conditional transportability is an uncertain, and even controversial, assumption.
\end{itemize}

For nested designs we need a slightly modified version of the positivity condition:
\begin{itemize}
    \item[A1$^*$.] Positivity: $\Pr[S=1|X=x] >0$ for all $x \in \mathcal{X}$ such that $f_{X}(x)>0$.
\end{itemize}
%That the prediction model is based on $X^*$ allows a larger set of covariates to be used to satisfy the transportability condition then is used to build the model.

%Would simply state the assumptions here. And give a remark that one is the critical untestable one. 

%Wonder if we should do this with the loss function having a subset of covariates in X. As far as I can see the theory is the same and might make this much more believable. 

%%%%%%%%%%%%%%%%%%%%%%%%%%%%%%%%%%%%%%%%%%%%%%%%%%%%%%%%%%%%%%%%%%%%%%%%%%%%%%
\subsection{Identification of measures of model performance}
%%%%%%%%%%%%%%%%%%%%%%%%%%%%%%%%%%%%%%%%%%%%%%%%%%%%%%%%%%%%%%%%%%%%%%%%%%%%%%

Under conditions A1 and A2, the risk in the target population for a non-nested design can be identified \cite{steingrimsson2023transporting} using a nested expectation (``g-formula''-like \cite{robins1986}) expression,
\[
\phi \equiv \E[\E[L(Y, h(X^*, \widehat \beta))|X, S=1]|S=0];
\]
or, equivalently, using an inverse odds weighting expression,
\[
\phi = \frac{1}{\Pr[S=0]} \E\left[\frac{I(S=1) \Pr[S=0|X]}{\Pr[S=1|X]} L(Y, h(X^*, \widehat \beta))\right].
\]
If conditions A1$^*$ and A2 hold, then the risk in the target population for a nested design can be identified \cite{steingrimsson2023transporting} using a different nested expectation expression,
\[
\psi \equiv \E[\E[L(Y, h(X^*, \widehat \beta))|X, S=1]];
\]
or, equivalently, using an inverse probability weighting expression, 
\[
\psi = \E\left[\frac{I(S=1)}{\Pr[S=1|X]} L(Y, h(X^*, \widehat \beta))\right].
\]
These identification results critically depend on assuming that the conditional transportability condition (A2) holds; in the rest of this manuscript, we consider methods for examining how sensitive results are to violations of this assumption.

%%%%%%%%%%%%%%%%%%%%%%%%%%%%%%%%%%%%%%%%%%%%%%%%%%%%%%%%%%%%%%%%%%%%%%%%%%%%%%
\section{Sensitivity analysis when transporting measures of model performance}\label{sec:sensitivity_analysis}
%%%%%%%%%%%%%%%%%%%%%%%%%%%%%%%%%%%%%%%%%%%%%%%%%%%%%%%%%%%%%%%%%%%%%%%%%%%%%%

\subsection{Sensitivity analysis model}

Assume that condition A2 does not hold, so that $ Y \mathlarger{\nindep} S | X, $ and
\[
f_{Y | X, S}(y | x,  s = 0) \neq f_{Y | X,S}(y | x, s = 1).
\]
We use an exponential tilt model \cite{scharfstein2018globalBiometrics, scharfstein2018globalSMMR, scharfstein2021global} to parameterize violations of the conditional transportability assumption. That is, we assume 
\begin{equation}\label{model_exponential_tilt}
	f_{Y | X, S}(y | x, s = 0) \propto e^{ \eta q(y)} f_{Y | X,  S}(y| x , s = 1), \eta \in \mathbb R,
\end{equation} 
where $q$ is a fixed increasing function and $\eta$ is the sensitivity analysis parameter (which is not identifiable because outcome information is unavailable from the target population). Setting $\eta = 0$ corresponds to the case where conditional transportability holds; $\eta$ values further from zero represent greater violations of the conditional transportability assumption. 

Because the left-hand-side of equation \eqref{model_exponential_tilt} is a density, we have that
\begin{equation}\label{model_exponential_tilt_expe}
  \begin{split}
  f_{Y | X, S}( y | x, s = 0) %&= \dfrac{f_{Y^a | X, S}(y| x, s = 1) e^{\psi_a q(y)}}{\mathlarger\int\limits_{y^\prime \in \mathcal Y^a} f_{Y^a| X, S}( y^\prime | x, s = 1) e^{\psi_a q(y^\prime)}dy^\prime} \\
      &= \dfrac{ e^{\eta q(y)} f_{Y | X, S}(y| x, s = 1) }{ \E [ e^{\eta q(Y)} | X = x, S = 1]}.
  \end{split}
\end{equation} 

For a binary outcome with $q$ as the identity function, the exponential tilt model implies that
\begin{equation}\label{model_exponential_tilt-bin}
	\Pr[Y=1 | X, S = 0] \propto e^{ \eta} \Pr[Y=1|X , S = 1], \eta \in \mathbb R;
\end{equation} 
it follows that $\eta >0$ ($\eta<0)$ implies that the conditional probability of the outcome in the target population is higher (lower) than in the source population.

\subsection{Relationship with selection models}

Using Bayes theorem, we can re-write the exponential tilt model in equation \eqref{model_exponential_tilt_expe} as
\begin{equation*}
	\begin{split}
	\dfrac{\Pr[S = 0 | X , Y = y] }{\Pr[S = 1 | X, Y = y] } &=  \dfrac{\Pr[S = 0 | X]}{\Pr[S= 1 | X]}  \times \dfrac{ e^{\eta q(y)}}{ \E [ e^{\eta q(Y)} | X, S = 1]}.
	\end{split}
\end{equation*} 
Taking logarithms gives,
\begin{equation}\label{model_selection_equiv_to_tilt}
	\begin{split}
	\mbox{logit} \big(\Pr[S = 0 | X , Y = y ]\big) &=  \mbox{logit} \big(\Pr[S = 0 | X]\big) + \eta q(y) - \ln \left(\E [ e^{\eta q(Y)} | X, S = 1 ]\right),
	\end{split}
\end{equation} 
where for a real number $0<u<1$, $\mbox{logit}(u) = \ln \big( u (1 -u)^{-1} \big) .$ In other words, the exponential tilt model has an interpretation as an odds of a selection model, where selection depends, in addition to the measured covariates $X$, on the outcome $Y$ which is not observed when $S = 0$.

%In Appendix \ref{appendix_selection_models}, we discuss some interesting consequences of the odds of selection parameterization and address the robustness of our estimators to model misspecification.

%%%%%%%%%%%%%%%%%%%%%%%%%%%%%%%%%%%%%%%%%%%%%%%%%%%%%%%%%%%%%%%%%%%%%%%%%%%%%%
\section{Identifiability, estimation, and inference}
\label{sec:id-sens}
%%%%%%%%%%%%%%%%%%%%%%%%%%%%%%%%%%%%%%%%%%%%%%%%%%%%%%%%%%%%%%%%%%%%%%%%%%%%%%

\subsection{Identifiability of the sensitivity analysis model}

In Appendix \ref{appendix:identification}, we show that, under our sensitivity analysis model, for a fixed $\eta$, the risk in the target population for a non-nested design is identified by
\begin{equation}\label{eq_mean_identification_S0}
  \phi(\eta) = \E\! \Bigg[ \dfrac{\E \big[ L(Y, h(X^*, \widehat \beta))e^{\eta q(Y)} | X, S = 1 \big]}{\E \big[ e^{\eta q(Y)} | X, S = 1 \big]} \Bigg| S = 0  \Bigg].
\end{equation}
Furthermore, in Appendix \ref{appendix:identification}, we show that, under our sensitivity analysis model, for a fixed $\eta$, the risk in the target population for a nested design is identified by
\begin{equation}\label{eq_mean_identification}
  \psi(\eta) = \E[S L(Y, h(X^*, \widehat \beta))] + \E \Bigg[I(S=0) \dfrac{\E \big[ L(Y, h(X^*, \widehat \beta))e^{\eta q(Y)} | X, S = 1 \big]}{\E \big[ e^{\eta q(Y)} | X, S = 1 \big]}\Bigg].
\end{equation}
The first term in expression \eqref{eq_mean_identification} is independent of the sensitivity parameter $\eta$ and represents the contribution of the sampled subset of the target population where the model is developed. Using data from participants from the source population does not rely on the sensitivity model (for them no assumption is in doubt) and thus their contribution to the overall analyses should not change with different values of the sensitivity parameter.

Because $\eta$ is not identifiable using the observed data $\mathcal{O}$ we propose to use expressions \eqref{eq_mean_identification_S0} and \eqref{eq_mean_identification} to conduct sensitivity analysis for a reasonable range of $\eta$ values; in Section \ref{sec:sel-par} we show how knowledge about the marginal probability of the outcome in the target population can be used to inform what range of $\eta$ values to consider.

\subsection{Estimation in the sensitivity analysis model for non-nested designs}

The sample analog of expression \eqref{eq_mean_identification_S0} gives the following conditional loss estimator for the risk in the target population under a non-nested design:
\[
\widehat \phi_{cl}(\eta) = \frac{1}{n_0}  \sum_{i=1}^{n} I(S_i=0) \widehat b(X_i; \eta),
\]
where $\widehat b(X; \eta)$ is an estimator for
\[
 b(X;\eta) =  \dfrac{\E \big[ L(Y, h(X^*, \widehat \beta))e^{\eta q(Y)} | X, S = 1 \big]}{\E \big[ e^{\eta q(Y)} | X, S = 1 \big]}.
\]
When $\eta=0$ (i.e.,~when the conditional transportability condition holds), the estimator $ \widehat \phi_{cl}(\eta)$ is equal to conditional loss estimator \cite{morrison2021}.

For binary $Y$ and $q$ as the identity function, we can estimate $b(X;\eta)$ using
\[
\widehat b(X; \eta) = \left(  \dfrac{  L(1, h(X^*, \widehat \beta)) e^{\eta} \widehat g(X) + L(0, h(X^*, \widehat \beta))  ( 1 - \widehat g(X))  }{  1 + \widehat g(X) (e^{\eta} - 1)  } \right),
\]
where $\widehat g(X)$ is an estimator for $\Pr[Y = 1 | X, S = 1]$. For a continuous $Y$, we can estimate $b(X;\eta)$ using a weighted linear regression of $L(Y, h(X^*, \widehat \beta))$ on $X$ using data from the source population with weights equal to $e^{\eta q(Y)}$. %We can rewrite the above equation as 
%\begin{equation}
% \widehat \phi_{cl}(\eta^\dagger) = \frac{1}{n_0}  \sum_{i=1}^{n} I(S_i=0) \left(  \dfrac{  L(1, h(X_i^*, \widehat \beta)) e^{\eta^\dagger} \widehat g(X_i) + L(0, h(X_i^*, \widehat \beta))  ( 1 - \widehat g(X_i))  }{  e^{\eta^\dagger} \widehat g(X_i) +  ( 1 - \widehat g(X_i))  } \right),
%\end{equation}
%where $\eta^\dagger = \eta(q(1) -q(0))$ and with $q$ as the idendity function we get $\eta = \eta^\dagger$.

In Supplementary Web Appendix \ref{appendix:influence_functions} we show that the influence function for $\phi(\eta)$ under the non-parametric model \cite{bickel1993efficient} of the observable data is
\begin{align*}
\Phi^1(\eta) &= \dfrac{1}{\Pr[S = 0]}  \vast\{ I(S=0) \Bigg\{ \dfrac{\E[L(Y, h(X^*, \widehat \beta))e^{\eta q(Y)} | X, S = 1]}{\E[e^{\eta q(Y)} | X, S = 1]}  - \phi(\eta) \Bigg\} \\
& + \dfrac{I(S = 1) \Pr[S = 0 | X] e^{\eta q(Y)} }{\Pr[S=1| X ] \E[e^{\eta q(Y)} | X, S = 1]} \times \Bigg\{ L(Y, h(X^*, \widehat \beta))-  \dfrac{\E[L(Y, h(X^*, \widehat \beta))e^{\eta q(Y)} | X, S = 1]}{\E[e^{\eta q(Y)} | X, S = 1]}   \Bigg\}   \vast\}.
\end{align*}
%It is interesting to note that under the sensitivity analysis model
%$$ \dfrac{\Pr[S = 0 | X] }{\Pr[S=1 | X ] } \times \dfrac{e^{\eta q(y)}}{\E[e^{\eta q(Y)} | X, S = 1]} = \dfrac{\Pr[S = 0 | X, Y]}{\Pr[S = 1 | X, Y]} $$
%this component is equal to the odds ratio for trial participation conditional on $X$ and $Y$. When condition A2 holds the corresponding term in the augmented estimator is 
%\[
%\dfrac{\Pr[S = 0 | X ]}{\Pr[S = 1 | X]}.
%\]

The influence function $\Phi^1(\eta)$ suggests the augmented estimator 
\begin{align*}
\widehat \phi_{aug}(\eta) = \dfrac{1}{n_0} \sum_{i=1}^n \left( I(S_i=0) \widehat b(X_i; \eta) + \dfrac{I(S_i = 1) (1 - \widehat p(X_i)) e^{\eta q(Y_i)} }{\widehat p(X_i) \widehat c(X_i;\eta) } \times \left( L(Y_i, h(X_i^*, \widehat \beta))-  \widehat b(X_i; \eta)\right) \right),
\end{align*}
where $\widehat p(X)$ is an estimator for $\Pr[S=1|X]$ and $\widehat c(X;\eta)$ is an estimator for $\E[e^{\eta q(Y)} | X, S = 1]$. The conditional loss estimator $\widehat \phi_{cl}(\eta)$ is a special case of the augmented estimator with $\widehat p(X_i) =1$ for all $i$. And if the conditional transportability condition holds (i.e.,~$\eta = 0$), then the augmented estimator is identical to the doubly robust estimator developed in  \cite{morrison2021}.  %The role of the estimators $\widehat b(X; \eta)$, $\widehat c(X;\eta)$, and $\widehat p(X)$ is to adjust for measured prediction error modifiers while the role of the sensitivity parameter $\eta$ is to examine the impact of potential unmeasured prediction error modifiers.

To study the large sample properties of $\widehat \phi_{aug}(\eta)$ we define, for arbitrary functions $b'(X)$, $c'(X)$, $p'(X)$, and $\gamma'$, the function
\begin{align*}
H(X,&S,Y; b'(X), c'(X), p'(X), \gamma') \\ &= \gamma' \left(I(S=0) b'(X) +  \frac{I(S=1) (1 - p'(X)) e^{\eta q(Y)}\left(L(Y, h(X^*, \widehat \beta)) - b'(X) \right) }{p'(X) c'(X)} \right).
\end{align*}
For a random variable $W$ we define $\mathbb{P}_n(W) = \frac{1}{n} \sum_{i=1}^n W_i$ and $\mathbb{G}_n(W) = \sqrt{n}(\mathbb{P}_n(W) - \E[W])$. Using this notation, we can write the augmented estimator as
\[
\widehat \phi_{aug}(\eta) = \mathbb{P}_n(H(X,S,Y; \widehat b(X; \eta), \widehat c(X;\eta), \widehat p(X), n/n_0)).
\]
In Supplementary Web Appendix \ref{As-Rep-Aug} we prove the following theorem about the large-sample properties of the augmented estimator:
\begin{theorem}
\label{thm-as-rep}
Let $b^*(X;\eta)$, $c^*(X;\eta)$, and $p^*(X)$ be the asymptotic limits of $\widehat b(X; \eta)$, $\widehat c(X;\eta)$, and $\widehat p(X)$, respectively. Under conditions B1-B5 listed in Supplementary Web Appendix \ref{As-Rep-Aug}, the augmented estimator $\widehat \phi_{aug}(\eta)$
\begin{itemize}
    \item Is consistent, that is, $\widehat \phi_{aug}(\eta) \overset{P}{\longrightarrow} \phi(\eta)$;
    \item Has the asymptotic representation
    \begin{equation}
    \label{thm-1-as-main}
    \sqrt{n} (\widehat \phi_{aug}(\eta) - \phi(\eta)) = \mathbb{G}_n(H(X,S,Y; b^*(X;\eta), c^*(X;\eta), p^*(X), \Pr[S=0]^{-1})) + Rem + o_P(1),
    \end{equation}
where the reminder term satisfies
\begin{align*}
Rem  \leq O_P\Bigg(1 &+ \sqrt{n}  \Bigg|\Bigg|\frac{\E[L(Y, h(X^*, \widehat \beta))e^{\eta q(Y)}|X,S=1]}{\E[e^{\eta q(Y)}|X,S=1]}-  \widehat b(X; \eta)\Bigg|\Bigg|^2_2 \times ||\Pr[S=1|X] - \widehat p(X)||^2_2 \\
&+ \sqrt{n} \Bigg|\Bigg|\frac{\E[L(Y, h(X^*, \widehat \beta))e^{\eta q(Y)}|X,S=1]}{\E[e^{\eta q(Y)}|X,S=1]}-  \widehat b(X; \eta)\Bigg|\Bigg|^2_2 \times ||\E[e^{\eta q(Y)}|X,S=1]-  \widehat c(X;\eta)||^2_2\Bigg).
\end{align*}
\end{itemize}
\end{theorem}
Theorem \ref{thm-as-rep} shows that the augmented estimator has the rate of convergence
\begin{align*}
\sqrt{n} (\widehat \phi_{aug}(\eta) - \phi(\eta)) &\leq O_P\Bigg(1 +  \sqrt{n}  \Bigg|\Bigg|\frac{\E[L(Y, h(X^*, \widehat \beta))e^{\eta q(Y)}|X,S=1]}{\E[e^{\eta q(Y)}|X,S=1]}-  \widehat b(X; \eta)\Bigg|\Bigg|^2_2 \times ||\Pr[S=1|X] - \widehat p(X)||^2_2 \\
&+ \sqrt{n} \Bigg|\Bigg|\frac{\E[L(Y, h(X^*, \widehat \beta))e^{\eta q(Y)}|X,S=1]}{\E[e^{\eta q(Y)}|X,S=1]}-  \widehat b(X; \eta)\Bigg|\Bigg|^2_2 \times ||\E[e^{\eta q(Y)}|X,S=1]-  \widehat c(X;\eta)||^2_2\Bigg).
\end{align*}
This implies that if the combined rate of $\widehat b(X; \eta)$ and $\widehat p(X)$ is at least $\sqrt{n}$ and the combined rate of $\widehat b(X; \eta)$ and $\widehat c(X;\eta)$ is at least $\sqrt{n}$, then $\widehat \phi_{aug}(\eta)$ is $\sqrt{n}$ convergent. Thus, $\widehat \phi_{aug}(\eta)$ can be $\sqrt{n}$ convergent even if the estimators $\widehat b(X; \eta)$, $\widehat p(X)$, and $\widehat c(X;\eta)$ converge at a slower rate than $\sqrt{n}$ as long as the two conditions on the combined rate of convergence hold (rate robustness \cite{smucler2019unifying}). Hence, augmented estimator can be used with data-adaptive estimators (such as GAMs) while still allowing for asymptotically valid inference \cite{chernozhukov2018double}. In contrast, the conditional loss estimator inherits the rate of convergence of $\widehat b(X; \eta)$. 

Assumption $B1$ in Supplementary Web Appendix \ref{As-Rep-Aug} suggests that only one of $\widehat b(X; \eta)$ or $(\widehat p(X), \widehat c(X;\eta))$ need to be consistent in order for the augmented estimator to be consistent, but not both (model robustness). This result, however, is not as useful as it might appear as implementation of both $\widehat b(X; \eta)$ and $\widehat c(X;\eta)$ relies on specifying the relationship between the outcome and the covariates in the sample from the source population. For example, in the special case of a binary $Y$ and with $q$ as the identity function we can write $\widehat c(X;\eta) = e^{\eta } \widehat g(X) + 1 - \widehat g(X)$ and the augmented estimator as
\begin{equation}
\label{aug-bin}
\widehat \phi_{aug}(\eta) = \frac{1}{n_0} \sum_{i=1}^n \left(I(S_i=0)\widehat b(X_i; \eta) +  \frac{I(S_i=1) (1 - \widehat p(X_i)) e^{\eta Y_i} }{\widehat p(X_i) (e^{\eta } \widehat g(X_i) +  1 - \widehat g(X_i) )} \left(L(Y_i, h(X_i^*, \widehat \beta)) - \widehat b(X_i; \eta) \right)\right),
\end{equation}
with
\begin{equation}
    \label{g-b}
\widehat b(X; \eta) = \left(  \dfrac{  L(1, h(X^*, \widehat \beta)) e^{\eta} \widehat g(X) + L(0, h(X^*, \widehat \beta))  ( 1 - \widehat g(X))  }{  e^{\eta} \widehat g(X) + 1 - \widehat g(X)  } \right).
\end{equation}
So both $\widehat b(X; \eta)$ and $\widehat c(X;\eta)$ rely on specifying $\widehat g(X)$. The following theorem, proved in Supplementary Web Appendix \ref{sec:DR}, shows that if $\widehat \phi_{aug}(\eta)$ is implemented using expression \eqref{aug-bin} with $b(X;\eta)$ estimated using expression \eqref{g-b}, then consistency relies on correctly specifying the model for $\Pr[Y=1|X, S=1]$ but does not rely on correctly specifying the model for $\Pr[S=1|X]$. 
\begin{theorem}
\label{non-nest-dr-2}
If $\widehat g(X) \overset{P}{\longrightarrow} \Pr[Y=1|X, S=1]$ and $\phi(\eta)$ is estimated using expression \eqref{aug-bin} with $b(X;\eta)$ estimated using expression \eqref{g-b}, then $\widehat \phi_{aug}(\eta) \overset{P}{\longrightarrow} \phi(\eta)$ whether the model for $\Pr[S=1|X]$ is correctly specified or not.
\end{theorem}

The augmented estimator $\widehat \phi_{aug}(\eta)$ (for binary, count, and continuous outcomes) has the advantage over the conditional loss estimator of being able to accommodate more flexible modeling of the nuisance parameters $\E[Y=1|X, S=1]$ and $\Pr[S=1|X]$ while allowing for asymptotically valid inference \cite{chernozhukov2018double}. This is important because subject matter knowledge is often inadequate for specifying parametric models for $\E[Y|X, S=1]$ and $\Pr[S=1|X]$. 

\subsection{Alternative parameterization of the selection model}

If we parameterize the sensitivity analysis using expression \eqref{model_selection_equiv_to_tilt}
\[
	\mbox{logit} \big(\Pr[S = 0 | X , Y = y ]\big) =  a(X;\eta) + \eta q(y),
\]
where $a(X;\eta) = \mbox{logit} \big(\Pr[S = 0 | X]\big) -  \ln (\E [ e^{\eta q(Y)} | X, S = 1 ])$, the augmented estimator can be written as
\begin{equation}
\label{aug-a}
\widehat \phi_{aug}(\eta) = \frac{1}{n_0} \sum_{i=1}^n \left(I(S_i=0)\widehat b(X_i; \eta) +  I(S_i=1) e^{\widehat a(X_i;\eta) + \eta q(Y_i)} \left(L(Y_i, h(X_i^*, \widehat \beta)) - \widehat b(X_i; \eta) \right)\right),
\end{equation}
where $\widehat a(X;\eta)$ is an estimator for $a(X;\eta)$. 

In addition to estimating $b(X;\eta)$, implementation of the augmented estimator using expression \eqref{aug-a} also requires estimating $a(X;\eta)$. To estimate $a(X;\eta)$ we start by noting that \cite{rotnitzky1998semiparametric, dahabreh2022global}
\begin{equation*}
\E\left[\frac{I(S=1)e^{a(X;\eta) + \eta q(Y)}}{\Pr[S=0]} - 1\right] = 0.
\end{equation*}
If we posit a parametric model $a(X, \theta; \eta)$ for $a(X;\eta)$ where $\theta$ is a finite dimensional parameter, then $\theta$ can be estimated using a generalized methods of moments estimator \cite{newey1994large}, which for a correctly specified parametric model is consistent under a mild conditions \cite{hansen1982large}. 

In Supplementary Web Appendix \ref{sec:DR} we prove the following doubly robustness property of $\widehat \phi_{aug}(\eta)$ when the alternative parameterization is used.
\begin{theorem}
\label{non-nest-dr-1}
If at least one of $\widehat b(X; \eta)$ or $\widehat a(X;\eta)$ are consistent, then the augmented estimator given by expression \eqref{aug-a} in non-nested designs is consistent (i.e.,~$\widehat \phi_{aug}(\eta) \overset{P}{\longrightarrow} \phi(\eta)$).
\end{theorem}
As the approach for estimating $\theta$ relies on specifying a parametric model for $a(X;\eta)$, the estimator cannot be used in combination with the more flexible modeling of $\widehat b(X; \eta)$ and $\widehat a(X;\eta)$.

\subsection{Estimation in the sensitivity analysis model for nested designs}

For nested designs, using the sample-analog of expression \eqref{eq_mean_identification} gives the conditional loss estimator for the risk in the target population under a nested design
\[
\widehat \psi_{cl}(\eta) = \frac{1}{n} \sum_{i=1}^n I(S_i = 1) L(Y_i, h(X^*_i, \widehat \beta))  + \frac{1}{n}  \sum_{i=1}^{n} I(S_i=0) \widehat b(X_i; \eta).
\]
The first term on the right hand side of the above equation represents contributions from observations from $S=1$ and the second term is a sum over observations with each term in the sum being an estimator for the conditional risk under our sensitivity analysis model.

In Supplementary Web Appendix \ref{appendix:influence_functions} we derive the non-parametric influence function under a nested design and the corresponding augmented estimator is
\begin{align}
\widehat \psi_{aug}(\eta) &= \frac{1}{n} \sum_{i=1}^n \Bigg( S_i L(Y_i, h(X_i^*, \widehat \beta)) \nonumber \\  &+ I(S_i=0) \widehat b(X_i; \eta) + \dfrac{I(S_i = 1) (1 - \widehat p(X_i)) e^{\eta q(Y_i)} \left(L(Y_i, h(X_i^*, \widehat \beta)) - \widehat b(X_i; \eta) \right)}{\widehat p(X_i) \widehat c(X_i;\eta)} \Bigg). \label{aug-nest}
\end{align}
In Supplementary Web Appendix \ref{app-nest} we show that the augmented estimator for nested designs has the same robustness and rate of convergence properties as the augmented estimator for non-nested designs.

%\subsection{Inference}

%For construction of Wald-type confidence intervals we can use the non-parametric bootstrap or sandwich variance-based estimators.

%%%%%%%%%%%%%%%%%%%%%%%%%%%%%%%%%%%%%%%%%%%%%%%%%%%%%%%%%%%%%%%%%%%%%%%%%%%%%%
\section{Selecting sensitivity parameter values using external information}
\label{sec:sel-par}
%%%%%%%%%%%%%%%%%%%%%%%%%%%%%%%%%%%%%%%%%%%%%%%%%%%%%%%%%%%%%%%%%%%%%%%%%%%%%%

When performing sensitivity analysis, selection of an appropriate range of the sensitivity parameter can be challenging. Now we show how we can base this choice on background knowledge of the marginal probability of the outcome in the target population.

Suppose that on the basis of substantive knowledge or prior research the expectation of the outcome in the target population $\E[Y | S = 0] = \mu$  is known. Then, even without individual-level outcome information from the target population sample, we may estimate $\eta$ as the solution to the sample analog of the following population estimating equation: 
\begin{equation}
  \int \int y  \dfrac{ e^{\eta q(y)} f_{Y | X, S}(y| x, s = 1) }{ \E [ e^{\eta q(Y)} | X = x, S = 1]} f_{X|S}(x | s = 0) dy dx - \mu = 0.
\end{equation}

For example, for binary outcome $Y$ and setting $q$ as the identity function, we can search for the $\eta$ value that solves
\begin{equation}
\label{exp-1}
 \sum_{i=1}^{n} I(S_i=0) \left\{  \dfrac{  e^{\eta} \widehat g(X_i) }{  e^{\eta } \widehat g(X_i) +  \{ 1 - \widehat g (X_i) \}  } - \mu \right\} = 0.
\end{equation}
Of course, perfect knowledge of $\mu$ may not be available in practical applications, but if a good-enough approximation to the true value of $\mu$ can be obtained, it helps anchor the sensitivity analysis by selecting a reasonable range of $\eta$ values to explore in sensitivity analyses. 

For the nested design, if the marginal prevalence rate in the target population $\E[Y] = \alpha$, we may estimate $\eta$ as the solution to the sample analog of the following population estimating equation:
\begin{align*}
 % &\int \int y f(y|x) f(x) dydx  =  \int \int y (f(y|x, s=1) \Pr[S=1|X] + f(y|x, s=0) (1 - \Pr[S=1|X])) f(x) dy dx \\
  \int \int y f_{Y | X, S}(y| x, s = 1)  f_X(x)  \left(\Pr[S=1|X=x] +\dfrac{ e^{\eta q(y)} \Pr[S=0|X=x]}{ \E [ e^{\eta q(Y)} | X = x, S = 1]} \right)dy dx - \alpha = 0.   
\end{align*}

\section{Sensitivity analysis using the Coronary Artery Surgery Study data}

We used data from the Coronary Artery Surgery Study (CASS) to illustrate the sensitivity analysis methods. 

\noindent \textbf{Study design and data:} CASS \cite{killip1983coronary} included a randomized trial of treatments for stable coronary artery disease that was nested within a cohort study of trial-eligible individuals. Of a total of 2099 trial eligible individuals, 780 were included in the randomized trial and 1319 declined to be randomized and were included in an observational study. In the trial, participants were randomly assigned to coronary artery surgery plus medical therapy (hereafter referred to as the surgery arm) versus medical therapy alone; the same treatments were used among non-randomized participants. Risk stratification is commonly of interest in randomized trials; one approach is to develop a risk model using observational data and then apply the estimated risk model in the randomized trial (e.g., to examine heterogeneity of treatment effects over predicted risk \cite{kent2007limitations, dahabreh2023toward}). Here, to illustrate the methods, we used the 430 participants in the observational part of CASS that received surgery as the sample from the source population and used covariate data from the participants in the randomized component that received surgery as the sample from the target population. We used death within 10-years from study entry as the outcome if interest. For simplicity and as previous analysis of the same data has shown limited impact of adjusting for missing data \cite{dahabreh2018generalizing}, we restricted the analysis to participants with complete covariate information (368 randomized and 955 non-randomized). Table \ref{cass_baseline} shows the distribution of the baseline covariates stratified by randomization status.
\\
\\
\noindent \textbf{Implementation:} We randomly split the dataset of non-randomized participants into two disjoint and approximately equally sized datasets. The data from the first dataset were used to fit the prediction model $h(X^*, \beta)$, a logistic regression model that included all the covariates listed in Table \ref{cass_baseline} as linear main effects (on the logit scale). The data from the second dataset (here used as the source population sample) were combined with covariate information from the randomized participants (here used as the sample from the target population) and the combined dataset was used to estimate the Brier risk in the target population using expression \eqref{aug-bin}. The vector of covariates $X$, used to satisfy the conditional transportability condition, included all the covariates listed in Table \ref{cass_baseline} (i.e.,~in this analysis we set $X = X^*$). The nuisance function $c(X;\eta)$ was estimated using the formula $\widehat c(X;\eta) = e^{\eta} \widehat g(X) + ( 1 - \widehat g(X))$ and $\widehat b(X;\eta)$ was estimated using expression \eqref{g-b}. The estimators for $\Pr[Y = 1 | X, S = 1]$ and $\Pr[S=1|X]$ needed to obtain $\widehat c(X;\eta)$ and $\widehat b(X;\eta)$ were based on logistic regression models with linear main effects of the variables listed in Table \ref{cass_baseline} and we used the non-parametric bootstrap with 1,000 bootstrap replicates to estimate Wald-style $95\%$ point-wise confidence intervals. 

The 10-year risk (cumulative incidence proportion) among non-randomized participants was $\widehat \mu = 0.186$ and we selected $\eta$ based on the range of the prevalence rate $[\widehat \mu/2, 2 \widehat \mu] = [0.093, 0.372]$ using expression \eqref{exp-1} and we implemented the augmented estimator using $\eta$ values in the corresponding range with increments of $0.05$ (resulting in a range of $\eta$ from [$-0.95,1.25$]). As a stability analysis, we i) estimated $\Pr[Y = 1 | X, S = 1]$ and $\Pr[S=1|X]$ using generalized additive models with the main effects of all continuous covariates (age and ejection fraction) modeled using B-splines (basis splines) and ii) used the jackknife to estimate the $95\%$ point-wise confidence intervals, with the same range of $\eta$ as the main analysis.
\\
\\
\noindent \textbf{Results:} Figure \ref{CassGlm} shows estimates of the Brier risk in the target population and associated 95\% confidence intervals over a range of values of $\eta$ when the non-parametric bootstrap was used for confidence interval construction and $\Pr[Y = 1 | X, S = 1]$ and $\Pr[S=1|X]$ were estimated using linear main effects logistic regression models. The estimates for the Brier risk ranged from 0.11 to 0.27.

Figures \ref{CassGlmJK}, \ref{GamBS}, and \ref{GamJK} in Supplementary Web Appendix \ref{CASS:SWA} show the results of the stability analysis using a generalized additive model to estimate both $\Pr[Y = 1 | X, S = 1]$ and $\Pr[S=1|X]$ and using the jackknife to construct confidence intervals. Generalized additive models produced results similar to those produced by the logistic regression models; the jackknife resulted in slightly narrower confidence intervals compared to the non-parametric bootstrap.

\clearpage

\begin{table}[ht!]
\caption{Baseline characteristics in CASS (August 1975 to December 1996). $S=0$ indicates  randomized participants that received surgery and $S=1$ indicates people in the observational component of CASS that received surgery.}
\label{cass_baseline}
\centering
\small
\begin{tabular}{@{}lll@{}}
\toprule
                     & $S=0$            & $S=1$                      \\ \midrule
Number of patients                    & 337            & 430             \\
Age      & 51.42 (7.21)   & 51.29 (7.68)               \\
History of angina      & 264 (78.3)     & 360 (83.7)         \\
Taken beta-blocker regularly        & 152 (45.1)     & 241 (56.0)        \\
Taken diuretic regularly     & 58 (17.2)      & 60 (14.0)         \\
Ejection fraction   & 60.61 (12.78)  & 60.20 (11.96)      \\
Employed full-time     & 237 (70.3)     & 286 (66.5)          \\
Type of job           &                &                     \\
\hspace{3mm} High physical labor job                   & 136 (40.4)     & 158 (36.7)       \\
\hspace{3mm} Low mental labor job                    & 119 (35.3)     & 134 (31.2)        \\
\hspace{3mm} High mental labor job                   & 82 (24.3)      & 138 (32.1)       \\
Left ventricular wall score   & 7.42 (2.83)    & 7.08 (2.73)       \\
Taken nitrates regularly      & 194 (57.6)     & 253 (58.8)         \\
History of MI     & 189 (56.1)     & 236 (54.9)      \\
Female         & 34 (10.1)      & 39 (9.1)           \\
Smoking status          &                &                     \\
\hspace{3mm} Never smoked                    & 56 (16.6)      & 76 (17.7)          \\
\hspace{3mm} Former smoker                    & 148 (43.9)     & 211 (49.1)    \\
\hspace{3mm} Current smoker                    & 133 (39.5)     & 143 (33.3)         \\
High limitation of activities   & 157 (46.6)     & 201 (46.7)        \\
High recreational activity    & 207 (61.4)     & 285 (66.3)        \\
Confirmed hypertension   & 104 (30.9)     & 109 (25.3)        \\
Confirmed diabetes & 0.03 (0.16)    & 0.03 (0.17)        \\
LMCA percent obstruction   & 3.66 (10.80)   & 7.95 (17.16)      \\
PLMA percent obstruction & 36.35 (38.26)  & 48.18 (39.81)      \\
Any diseased proximal vessels   & 205 (60.8)     & 317 (73.7)         \\
Systolic blood pressure  & 129.13 (17.85) & 129.85 (17.83)     \\ \bottomrule
\end{tabular}
\caption*{LMCA = left main coronary artery; MI = myocardial infarction; PLMA = proximal left anterior artery. For continuous variables we report the mean (standard deviation); for binary variables we report the number of patients (percentage).}
\end{table}

\clearpage 
\vspace{0.3in}
% truth is 0.2595, relative bias is: - 12.19%, (0.2278638 - 0.2594847)/0.2594847
\begin{figure}[ht!]
    \centering
    \includegraphics[width=0.8\textwidth]{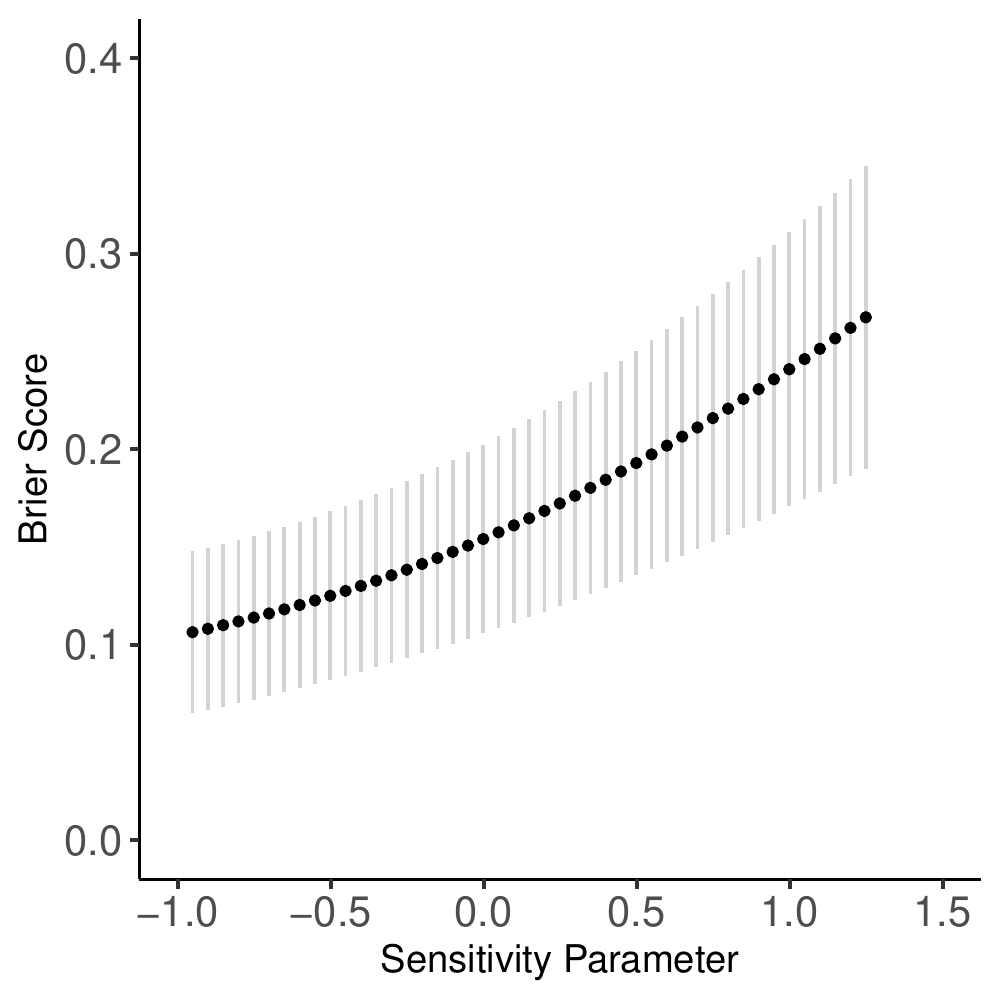}
    \caption{Sensitivity analysis using CASS data. The values used for the sensitivity parameter ($\eta$) are on the x-axis and the corresponding estimates of the Brier risk calculated using the augmented estimator for a non-nested design are on the y-axis ($\widehat \phi_{aug}(\eta)$). The nuisance functions $\Pr[Y = 1 | X, S = 1]$ and $\Pr[S=1|X]$ were estimated using logistic regression models and 95\% confidence intervals were calculated using the non-parametric bootstrap with $1,000$ bootstrap replicates. The solid line connects point estimates and the gray lines are point-wise 95\% confidence intervals.}
    \label{CassGlm}
\end{figure}

\section{Discussion}

We considered the problem of estimating model performance in a target population that differs from the source population used for model development or model evaluation, when information on covariates, but not outcomes, is available from the target population. In much of the literature studying this setting, methods for tailoring prediction models and for evaluating performance in the target population rely on a conditional exchangeability assumption that the available covariates are sufficient to render the outcome independent of population (source or target population). When subject matter knowledge is insufficient to determine whether the assumption is plausible, analysts need to evaluate how its violations would impact the findings. Here, we developed a global sensitivity analysis approach to violations of the conditional exchangeability condition using an exponential tilt model. We derived two sensitivity analysis estimators: a plug-in estimator and an augmented estimator that is obtained from the non-parametric influence function under the sensitivity analysis model. We suggested an approach for selecting a reasonable range of values for the sensitivity parameters based on background knowledge about the prevalence rate in the target population. Last, we applied the methods to data on individuals with stable ischemic heart disease undergoing coronary revascularization surgery.

Our approach addresses a key limitation of methods for transporting prediction models and assessing their performance in a target population. Future research could address issues such as missing data (other than the outcome data in the target population), failure-time outcomes \cite{steingrimsson2022extending}, and measurement error, or extensions to more complex measures of model performance such as the area under the receiver operating characteristic curve \cite{li2022estimating}.

The methods proposed here have the advantage of relating the sensitivity parameter to the marginal probability of the outcome in the target population, allowing investigators to choose an initial sensitivity parameter value by drawing on background knowledge about the target population. Depending on how sharp this knowledge is, analysts can choose to expand the sensitivity analysis to cover an appropriately dispersed set of additional sensitivity parameter values around the one implied by the postulated marginal probability in the target population.

%%%%%%%%%%%%%%%%%%%%%%%%%%%%%%%%%%%%%%%%%%%%%%%%%%%%%%%%%%%%%%%%%%%%%%%%%%%%%%
% BIBLIOGRAPHY
%%%%%%%%%%%%%%%%%%%%%%%%%%%%%%%%%%%%%%%%%%%%%%%%%%%%%%%%%%%%%%%%%%%%%%%%%%%%%%
\clearpage
\bibliographystyle{ieeetr}
\bibliography{references}
%%%%%%%%%%%%%%%%%%%%%%%%%%%%%%%%%%%%%%%%%%%%%%%%%%%%%%%%%%%%%%%%%%%%%%%%%%%%%%

%%%%%%%%%%%%%%%%%%%%%%%%%%%%%%%%%%%%%%%%%%%%%%%%%%%%%%%%%%%%%%%%%%%%
% \section{Acknowledgments}
%%%%%%%%%%%%%%%%%%%%%%%%%%%%%%%%%%%%%%%%%%%%%%%%%%%%%%%%%%%%%%%%%%%%
%This work was supported by Patient-Centered Outcomes Research Institute (PCORI) awards ME-1306-03758 and ME-1502-27794 and National Institutes of Health (NIH) grant R37 AI102634. Statements in this paper do not necessarily represent the views of the PCORI, its Board of Governors, the Methodology Committee, or the NIH. The data analyses in our paper used CASS research materials obtained from the NHLBI Biologic Specimen and Data Repository Information Coordinating Center. This paper does not necessarily reflect the opinions or views of the CASS or the NHLBI. 

%%%%%%%%%%%%%%%%%%%%%%%%%%%%%%%%%%%%%%%%%%%%%%%%%%%%%%%%%%%%%%%%%%%%%%%%%%%%%%
\appendix
\clearpage
%%%%%%%%%%%%%%%%%%%%%%%%%%%%%%%%%%%%%%%%%%%%%%%%%%%%%%%%%%%%%%%%%%%%%%%%%%%%%%

%%%%%%%%%%%%%%%%%%%%%%%%%%%%%%%%%%%%%%%%%%%%%%%%%%%%%%%%%%%%%%%%%%%%%%%%%%%%%%
\section{Identification of measures of loss-based loss performance under the sensitivity analysis model}\label{appendix:identification}
\renewcommand{\theequation}{\thesection.\arabic{equation}}
\setcounter{equation}{0}
%%%%%%%%%%%%%%%%%%%%%%%%%%%%%%%%%%%%%%%%%%%%%%%%%%%%%%%%%%%%%%%%%%%%%%%%%%%%%%

Under the sensitivity analysis model in Section \ref{sec:sensitivity_analysis} of the main text, the outcome density in the target population can be written as
\begin{equation*}
  f_{Y | X, S}(y | x, s = 0) =  \frac{e^{ \eta q(y)} f_{Y | X,  S}(y| x , s = 1)}{ \E \big[ e^{\eta q(Y)} | X = x, S = 1 \big] }, \eta \in \mathbb R,
\end{equation*} 
where $q$ is an unspecified deterministic increasing function.
This implies that,
\begin{equation}
\E[L(Y, h(X^*, \widehat \beta)) | X, S = 0] = \dfrac{\E \big[ L(Y, h(X^*, \widehat \beta)) e^{\eta q(Y)} | X, S = 1 \big]}{\E \big[ e^{\eta q(Y)} | X, S = 1\big]}. \label{rel-S}
\end{equation}
For non-nested designs the relationship in equation \eqref{rel-S} gives that
\begin{equation*}
  \begin{split}
\E[L(Y, h(X^*, \widehat \beta))| S = 0] &= \E\! \big[\E[L(Y, h(X^*, \widehat \beta))| X, S = 0 ] \big| S = 0\big] \\
&= \E\! \Bigg[ \dfrac{\E \big[ L(Y, h(X^*, \widehat \beta)) e^{\eta q(Y)} | X, S = 1 \big]}{\E \big[ e^{\eta q(Y)} | X, S = 1 \big]} \Bigg| S = 0  \Bigg].
  \end{split}
\end{equation*}

For nested designs we have
\begin{equation*}
  \begin{split}
&\E[L(Y, h(X^*, \widehat \beta))] \\ &= \E \big[\E[L(Y, h(X^*, \widehat \beta))| X]\big] \\
&= \E\! \big[\Pr[S=1|X] \E[L(Y, h(X^*, \widehat \beta)| X, S=1]  +  \Pr[S=0|X] \E[L(Y, h(X^*, \widehat \beta))| X, S=0] \big] \\
&= \E\! \big[S \E[L(Y, h(X^*, \widehat \beta))| X, S=1]  + I(S=0) \E[L(Y, h(X^*, \widehat \beta))| X, S=0]  \big] \\
&= \E\! \Bigg[S \E[L(Y, h(X^*, \widehat \beta))| X, S=1]  +  I(S=0) \dfrac{\E \big[ L(Y, h(X^*, \widehat \beta))e^{\eta q(Y)} | X, S = 1 \big]}{\E \big[ e^{\eta q(Y)} | X, S = 1 \big]}  \Bigg].
  \end{split}
\end{equation*}
Noting that 
\begin{equation} \label{eq:re_expression_of_term}
  \begin{split}
\E \big[ S \E [L(Y, h(X^*, \widehat \beta))| X, S =1] \big] &= \E \big[ S \E [L(Y, h(X^*, \widehat \beta))| X, S ] \big] \\
          &= \E \big[  \E [ S L(Y, h(X^*, \widehat \beta))| X, S ] \big] \\
          &= \E \big[ S L(Y, h(X^*, \widehat \beta))\big],
  \end{split}
\end{equation}
we conclude that
\begin{equation*}
  \begin{split}
\E[L(Y, h(X^*, \widehat \beta))] = \E \big[ S L(Y, h(X^*, \widehat \beta))\big] + \E\! \Bigg[ I(S=0) \dfrac{\E \big[ L(Y, h(X^*, \widehat \beta))e^{\eta q(Y)} | X, S = 1 \big]}{\E \big[ e^{\eta q(Y)} | X, S = 1 \big]}  \Bigg].
  \end{split}
\end{equation*}

\clearpage
\section{Influence functions}\label{appendix:influence_functions}

\subsection*{Non-nested designs}

We use the pathwise derivative of $\phi$ to calculate the efficient influence function under the non-parametric model for the observable data \cite{van2000asymptotic}. Let $\frac{d}{dt} \phi_t$ be the derivative of $\phi_t$ under the parametric submodel $p_t, t \in [0,1[$ with $t=0$ denoting the true data law. Let $g(\cdot)$ denote a general score function. We have
\begin{align*}
\frac{d}{dt} \phi_t(\eta)\Big|_{t=0} &= \frac{d}{dt} \E_{p_t}\left[\frac{\E_{p_t}[L(Y, h(X^*, \widehat \beta)) e^{\eta q(Y)}|X, S=1]}{\E_{p_t}[e^{\eta q(Y)}|X, S=1]} \Bigg | S=0 \right]\Bigg|_{t=0} \\
&=  \underbrace{\E_{p_0}\left[\frac{d}{dt} \frac{\E_{p_t}[L(Y, h(X^*, \widehat \beta)) e^{\eta q(Y)}|X, S=1]}{\E_{p_t}[e^{\eta q(Y)}|X, S=1]} \Bigg|_{t=0}\Bigg | S=0 \right]}_{(i)}
\\ &+ \underbrace{\frac{d}{dt} \E_{p_t}\left[ \frac{\E_{p_0}[L(Y, h(X^*, \widehat \beta)) e^{\eta q(Y)}|X, S=1]}{\E_{p_0}[e^{\eta q(Y)}|X, S=1]} \Bigg | S=0 \right]\Bigg|_{t=0}}_{(ii)}
\end{align*}
Start by looking at term $(ii)$, 
\begin{align*}
(ii) &= \frac{d}{dt} \E_{p_t}\left[ \frac{\E_{p_0}[L(Y, h(X^*, \widehat \beta)) e^{\eta q(Y)}|X, S=1]}{\E_{p_0}[e^{\eta q(Y)}|X, S=1]} \Bigg | S=0 \right]\Bigg|_{t=0} \\ &= \E_{p_0}\left[ \frac{\E_{p_0}[L(Y, h(X^*, \widehat \beta)) e^{\eta q(Y)}|X, S=1]}{\E_{p_0}[e^{\eta q(Y)}|X, S=1]} g_{Y,X|S=0}\Bigg | S=0 \right].
\end{align*}
Define 
\[
b(X) = \frac{\E_{p_0}[L(Y, h(X^*, \widehat \beta)) e^{\eta q(Y)}|X, S=1]}{\E_{p_0}[e^{\eta q(Y)}|X, S=1]}.
\]
Using this notation 
\begin{align*}
(ii) &=  \E_{p_0}\left[ b(X) g_{Y,X|S=0}\Bigg | S=0 \right]\\
&=\E_{p_0}\left[ \left(b(X) - \E_{p_0}[b(X)|S=0]\right) g_{Y,X|S=0}\Bigg | S=0 \right]\\
&= \E_{p_0}\left[\frac{I(S=0)}{\Pr[S=0]} \left(b(X) - \E_{p_0}[b(X)|S=0]\right) g_{Y,X,S} \right].
\end{align*}
By equation \ref{rel-S} 
\begin{align*}
\E_{p_0}[b(X)|S=0] &= \E_{p_0} \left[\frac{\E_{p_0}[L(Y, h(X^*, \widehat \beta)) e^{\eta q(Y)}|X, S=1]}{\E_{p_0}[e^{\eta q(Y)}|X, S=1]} \Big|S=0\right]\\
&= \E_{p_0}[\E_{p_0}[L(Y, h(X^*, \widehat \beta)|X, S=0]|S=0] \\
&= \E_{p_0}[L(Y, h(X^*, \widehat \beta)|S=0]\\
&= \phi.
\end{align*}
Thus,
\[
(ii) = \E_{p_0}\left[\frac{I(S=0)}{\Pr[S=0]} \left(b(X) - \phi \right) g_{Y,X,S} \right].
\]
Now look at the first term $(i)$. We have
\begin{align*}
&\frac{d}{dt} \E_{p_t} [L(Y, h(X^*, \widehat \beta))e^{\eta q(Y)}|X,S=1] \\ &=\E_{p_0}[L(Y, h(X^*, \widehat \beta))e^{\eta q(Y)} g_{Y|X,S=1}|X,S=1] \\
&=\E_{p_0}[(L(Y, h(X^*, \widehat \beta))e^{\eta q(Y)} - \E_{p_0}[L(Y, h(X^*, \widehat \beta))e^{\eta q(Y)}|X, S=1]) g_{Y|X,S=1}|X,S=1] \end{align*}
and
\[
\frac{d}{dt} \E_{p_t} [e^{\eta q(Y)}|X,S=1] =\E_{p_0}[(e^{\eta q(Y)} - \E_{p_0}[e^{\eta q(Y)}|X, S=1]) g_{Y|X,S=1}|X,S=1].
\]
Using this
\begin{align*}
& \E_{p_0}\left[ \frac{d}{dt} \frac{\E_{p_t}[L(Y, h(X^*, \widehat \beta)) e^{\eta q(Y)}|X, S=1]}{\E_{p_t}[e^{\eta q(Y)}|X, S=1]}\Bigg|_{t=0} \Bigg | S=0 \right] \\
&= \underbrace{\E_{p_0}\left[ \frac{\E_{p_0} [(L(Y, h(X^*, \widehat \beta))e^{\eta q(Y)} - \E_{p_0}[L(Y, h(X^*, \widehat \beta)) e^{\eta q(Y)}|X, S=1]) g_{Y|X,S=1}|X,S=1] }{\E_{p_0}[e^{\eta q(Y)}|X, S=1]} \Bigg | S=0 \right]}_{(iii)} \\
&- \underbrace{E_{p_0}\left[ \frac{\E_{p_0}[L(Y, h(X^*, \widehat \beta)) e^{\eta q(Y)}|X, S=1]}{\E_{p_0}[e^{\eta q(Y)}|X, S=1]^2 } \E_{p_0}[(e^{\eta q(Y)} - \E_{p_0}[e^{\eta q(Y)}|X, S=1]) g_{Y|X,S=1}|X,S=1] \Bigg | S=0 \right]}_{(iv)}
\end{align*}
Start with term $(iii)$
\begin{align*}
&\hspace{-25pt} (iii) = \E_{p_0} \left[ \frac{\E_{p_0} [(L(Y, h(X^*, \widehat \beta)) e^{\eta q(Y)} - \E_{p_0}[L(Y, h(X^*, \widehat \beta)) e^{\eta q(Y)}|X, S=1]) g_{Y|X,S=1}|X,S=1] }{\E_{p_0}[e^{\eta q(Y)}|X, S=1]} \Bigg | S=0 \right] \\ 
&\hspace{-25pt}= \E_{p_0} \left[ \frac{I(S=0)}{\Pr[S=0]} \frac{\E_{p_0} [(L(Y, h(X^*, \widehat \beta)) e^{\eta q(Y)} - \E_{p_0}[L(Y, h(X^*, \widehat \beta)) e^{\eta q(Y)}|X, S=1]) g_{Y|X,S=1}|X,S=1] }{\E_{p_0}[e^{\eta q(Y)}|X, S=1]}\right] \\ 
&\hspace{-45pt}= \E_{p_0} \left[ \E_{p_0} \left[ \frac{I(S=0)}{\Pr[S=0]} \frac{\E_{p_0} \left[\frac{I(S=1)}{\Pr[S=1|X]} (L(Y, h(X^*, \widehat \beta)) e^{\eta q(Y)} - \E_{p_0}[L(Y, h(X^*, \widehat \beta))e^{\eta q(Y)}|X, S=1]) g_{Y|X,S=1}|X \right] }{\E_{p_0}[e^{\eta q(Y)}|X, S=1]}\Bigg | X \right] \right] \\ 
&\hspace{-25pt}= \E_{p_0} \left[  \frac{\Pr[S=0|X]}{\Pr[S=0]} \frac{\E_{p_0} \left[\frac{I(S=1)}{\Pr[S=1|X]} (L(Y, h(X^*, \widehat \beta))e^{\eta q(Y)} - \E_{p_0}[L(Y, h(X^*, \widehat \beta))e^{\eta q(Y)}|X, S=1]) g_{Y|X,S=1}|X \right] }{\E_{p_0}[e^{\eta q(Y)}|X, S=1]} \right] \\
&\hspace{-25pt}= \frac{\E_{p_0} \left[ \frac{I(S=1) \Pr[S=0|X]}{\Pr[S=1|X] \E_{p_0}[e^{\eta q(Y)}|X, S=1] }   (L(Y, h(X^*, \widehat \beta))e^{\eta q(Y)} - \E_{p_0}[L(Y, h(X^*, \widehat \beta))e^{\eta q(Y)}|X, S=1]) g_{Y|X,S=1}  \right]}{\Pr[S=0]} \\
&\hspace{-25pt}= \frac{\E_{p_0} \left[ \frac{I(S=1) \Pr[S=0|X]}{\Pr[S=1|X] \E_{p_0}[e^{\eta q(Y)}|X, S=1] }   (L(Y, h(X^*, \widehat \beta))e^{\eta q(Y)} - \E_{p_0}[L(Y, h(X^*, \widehat \beta))e^{\eta q(Y)}|X, S=1]) g_{Y,X,S}  \right]}{\Pr[S=0]}.
\end{align*}
Similarly for term $(iv)$ we have 
\begin{align*}
&(iv) = \frac{\E_{p_0} \left[\frac{I(S=1) \Pr[S=0|X] \E_{p_0}[L(Y, h(X^*, \widehat \beta))e^{\eta q(Y)}|X, S=1]}{\Pr[S=1|X] \E_{p_0}[e^{\eta q(Y)}|X, S=1]^2} (e^{\eta q(Y)} - \E_{p_0}[e^{\eta q(Y)}|X, S=1]) g_{Y,X,S} \right]}{\Pr[S=0]}.
\end{align*}
Hence, term (i) is equal to
\begin{align*}
&\hspace{-30pt}\frac{1}{\Pr[S=0]} \E_{p_0} \Bigg[ \frac{I(S=1) \Pr[S=0|X] }{\Pr[S=1|X] \E_{p_0}[e^{\eta q(Y)}|X, S=1]} \Bigg(L(Y, h(X^*, \widehat \beta))e^{\eta q(Y)} - \E_{p_0}[L(Y, h(X^*, \widehat \beta))e^{\eta q(Y)}|X, S=1] \\
&- \frac{\E_{p_0}[L(Y, h(X^*, \widehat \beta)) e^{\eta q(Y)}|X, S=1]}{\E_{p_0}[e^{\eta q(Y)}|X, S=1]} (e^{\eta q(Y)} - \E_{p_0}[ e^{\eta q(Y)}|X, S=1])  \Bigg) g_{Y,X,S}\Bigg]
\end{align*}

Combing the above gives that the influence function of the functional in (\ref{eq_mean_identification_S0}) is
\begin{equation*}
	\begin{split}
\mathit{\Phi}^1(\eta) &= \dfrac{1}{\Pr[S = 0]}  \vast\{ I(S=0) \Bigg\{ \dfrac{\E[L(Y, h(X^*, \widehat \beta))e^{\eta q(Y)} | X, S = 1]}{\E[e^{\eta q(Y)} | X, S = 1]}  - \phi\Bigg\} \\
&\quad\quad\quad + \dfrac{I(S = 1) \Pr[S = 0 | X]}{\Pr[S=1 | X ] \E[e^{\eta q(Y)} | X, S = 1]} \\
&\quad \times \Bigg\{ L(Y, h(X^*, \widehat \beta))e^{\eta q(Y)}  - \E[L(Y, h(X^*, \widehat \beta))e^{\eta q(Y)} | X, S = 1]  - \dfrac{\E[L(Y, h(X^*, \widehat \beta))e^{\eta q(Y)} | X, S = 1]}{\E[e^{\eta q(Y)} | X, S = 1]} \\
&\quad\quad\quad\quad\quad\quad\quad\quad\quad\quad\quad\quad\quad\quad\quad\quad \times \Big\{  e^{\eta q(Y)} - \E[e^{\eta q(Y)} | X, S = 1]  \Big\}   \Bigg\}   \vast\} \\ 
&= \dfrac{1}{\Pr[S = 0]}  \vast\{ I(S=0) \Bigg\{ \dfrac{\E[L(Y, h(X^*, \widehat \beta))e^{\eta q(Y)} | X, S = 1]}{\E[e^{\eta q(Y)} | X, S = 1]}  - \phi \Bigg\} \\
&\quad\quad\quad + \dfrac{I(S = 1) \Pr[S = 0 | X] e^{\eta q(Y)} }{\Pr[S=1| X ] \E[e^{\eta q(Y)} | X, S = 1]} \times \Bigg\{ L(Y, h(X^*, \widehat \beta))-  \dfrac{\E[L(Y, h(X^*, \widehat \beta))e^{\eta q(Y)} | X, S = 1]}{\E[e^{\eta q(Y)} | X, S = 1]}   \Bigg\}   \vast\} .
	\end{split}
\end{equation*}

\subsection*{Nested design}

Write the observed data sensitivity analysis functional in (\ref{eq_mean_identification}) as 
\begin{equation*}
  \begin{split}
  \psi(\eta) &=  \E \big[ S L(Y, h(X^*, \widehat \beta))\big]  +  \E \Bigg[  I(S=0) \dfrac{ \E[L(Y, h(X^*, \widehat \beta))e^{\eta q(Y)} | X, S = 1]  }{ \E [e^{\eta q(Y)} | X, S = 1]  }   \Bigg] \\
    &=  \underbrace{ \E \big[ S L(Y, h(X^*, \widehat \beta))\big]   \big]}_{\psi_1} + \underbrace{\E  \Bigg[  I(S=0) \dfrac{ \E[L(Y, h(X^*, \widehat \beta))e^{\eta q(Y)} | X, S = 1]  }{ \E[e^{\eta q(Y)} | X, S = 1]  }  \Bigg]}_{\psi_2(\eta)}.
  \end{split}
\end{equation*}

The influence function for the first term, $\psi_1$, is
\begin{equation*}
  \begin{split}
\mathit{\Psi}_1^1(\eta)  = S  L(Y, h(X^*, \widehat \beta))-  \E[ S L(Y, h(X^*, \widehat \beta))] = S  L(Y, h(X^*, \widehat \beta))- \psi_1(\eta). 
  \end{split}
\end{equation*}
%Of note, $\psi_1$ and $\mathit{\Psi}_1^1$ do not depend on the sensitivity parameters $\eta$; they represent the contribution of the individuals from the source population, for whom no assumption is in doubt and thus their contribution to the overall analyses should not change with different values of the sensitivity parameter.

Using similar calculations as for the non-nested design, the influence function for the second term, $\psi_2(\eta)$, is
\begin{equation*}
	\begin{split}
\mathit{\Psi}_2^1(\eta) &= I(S=0) \dfrac{\E[L(Y, h(X^*, \widehat \beta))e^{\eta q(Y)} | X, S = 1]}{\E[e^{\eta q(Y)} | X, S = 1]} + \dfrac{I(S = 1) \Pr[S = 0 | X]}{\Pr[S=1 | X ] \E[e^{\eta q(Y)} | X, S = 1]} \\
&\quad \times \Bigg\{ L(Y, h(X^*, \widehat \beta))e^{\eta q(Y)}  - \E[L(Y, h(X^*, \widehat \beta))e^{\eta q(Y)} | X, S = 1] \\
&\quad\quad\quad\quad\quad\quad\quad - \dfrac{\E[L(Y, h(X^*, \widehat \beta))e^{\eta q(Y)} | X, S = 1]}{\E[e^{\eta q(Y)} | X, S = 1]} \times \Big\{  e^{\eta q(Y)} - \E[e^{\eta q(Y)} | X, S = 1]  \Big\}   \Bigg\} \\
&\quad - \psi_{2}(\eta).
  \end{split}
\end{equation*}

With a bit of algebra, we obtain 
\begin{equation*}
  \begin{split}
\mathit{\Psi}_2^1(\eta) &= I(S=0) \dfrac{\E[L(Y, h(X^*, \widehat \beta))e^{\eta q(Y)} | X, S = 1]}{\E[e^{\eta q(Y)} | X, S = 1]} \\
&\quad + \dfrac{I(S = 1) \Pr[S = 0 | X] e^{\eta q(Y)}}{\Pr[S=1 | X ] \E[e^{\eta q(Y)} | X, S = 1]} \times \Bigg\{ L(Y, h(X^*, \widehat \beta)) - \dfrac{\E[L(Y, h(X^*, \widehat \beta))e^{\eta q(Y)} | X, S = 1]}{\E[e^{\eta q(Y)} | X, S = 1]}  \Bigg\} - \psi_{2}(\eta).
	\end{split}
\end{equation*}

The influence function for $\psi(\eta)$ is 
%\begin{equation*}
$$\mathit{\Psi}^1 (\eta) = \mathit{\Psi}_1^1 + \mathit{\Psi}_2^1(\eta).$$
%\end{equation*}

\section{Asymptotic properties of the augmented estimator for non-nested designs}
\label{sec:DR}

\subsection{Asymptotic representation of augmented estimator for non-nested designs}
\label{As-Rep-Aug}
Recall that the augmented estimator is given by
\begin{align*}
\widehat \phi_{aug}(\eta) = \dfrac{1}{n_0} \sum_{i=1}^n \left( I(S_i=0) \widehat b(X_i; \eta) + \dfrac{I(S_i = 1) (1 - \widehat p(X_i)) e^{\eta q(Y_i)} }{\widehat p(X_i) \widehat c(X_i;\eta) } \times \left( L(Y_i, h(X_i^*, \widehat \beta))-  \widehat b(X_i; \eta)\right) \right),
\end{align*}

We make the following assumptions 
\begin{enumerate}
\item[B1] At least one of 
\[
\widehat b(X; \eta) \overset{P}{\longrightarrow}   \dfrac{\E \big[ L(Y, h(X^*, \widehat \beta))e^{\eta q(Y)} | X, S = 1 \big]}{\E \big[ e^{\eta q(Y)} | X, S = 1 \big]}
\]
or
\[
(\widehat p(X), \widehat c(X;\eta)) \overset{P}{\longrightarrow}  (\Pr[S=1|X], \E[e^{\eta q(Y)}|X, S=1]).
\]
\item[B2] The sequences $H(X,S,Y; b^*(X;\eta), c^*(X;\eta), p^*(X),\Pr[S=0]^{-1})$ and $H(X,S,Y; \widehat b(X; \eta), \widehat c(X;\eta), \widehat p(X),n/n_0)$ are Donsker.
\item[B3] \[
||H(X,S,Y; \widehat b(X; \eta), \widehat c(X;\eta), \widehat p(X), n/n_0) - H(X,S,Y; b^*(X;\eta), c^*(X;\eta), p^*(X), \Pr[S=0]^{-1})||_2 \overset{P}{\longrightarrow} 0.
\]
\item[B4] $\E[H(X,S,Y; b^*(X;\eta), c^*(X;\eta), p^*(X), \Pr[S=0]^{-1})^2] < \infty$.
\item[B5] $\widehat p(X)$ and $\widehat c(X;\eta)$ are uniformly bounded away from zero and $\E[L(Y, h(X^*, \widehat \beta))e^{\eta q(Y)}|X,S=1]$, $\E[e^{\eta q(Y)}|X,S=1]$, $\widehat b(X; \eta)$ are uniformly bounded.
\end{enumerate}
Now we will prove Theorem \ref{thm-as-rep} in the main text.
\begin{proof} 
\textbf{Consistency:} 
We have that 
\[
\widehat \phi_{aug}(\eta) \overset{P}{\longrightarrow} \frac{1}{\Pr[S=0]} \E\left[ I(S=0) b^*(X;\eta) +  \frac{I(S=1) (1 - p^*(X)) e^{\eta q(Y)} }{p^*(X) c^*(X;\eta)} \left(L(Y, h(X^*, \widehat \beta)) - b^*(X;\eta) \right)\right].
\]
Now we show that the right hand side of the above equation is equal to $\phi(\eta)$ if either $\widehat b(X; \eta)$ is consistent or both  $\widehat p(X)$ and $\widehat c(X;\eta)$ are consistent.
\\
\\
\textbf{Case 1:} First assume that 
\[
\widehat b(X; \eta) \overset{P}{\longrightarrow}   \dfrac{\E \big[ L(Y, h(X^*, \widehat \beta))e^{\eta q(Y)} | X, S = 1 \big]}{\E \big[ e^{\eta q(Y)} | X, S = 1 \big]},
\]
but we make no assumptions on the limits $p^*(X)$ and $c^*(X;\eta)$. Using that assumption
\begin{align*}
\frac{1}{\Pr[S=0]} &\E\left[\frac{I(S=1) (1 - p^*(X)) e^{\eta q(Y)} }{p^*(X) c^*(X;\eta) } \left(L(Y, h(X^*, \widehat \beta)) - b^*(X;\eta)\right)\right] \\
&= \frac{1}{\Pr[S=0]} \E\left[\E\left[ \frac{I(S=1) (1 - p^*(X)) e^{\eta q(Y)} }{p^*(X) c^*(X;\eta) } \left(L(Y, h(X^*, \widehat \beta)) - b^*(X;\eta)\right)\Big|X \right]\right] \\
&= \frac{1}{\Pr[S=0]} \E\left[\E\left[ \frac{I(S=1) (1 - p^*(X)) e^{\eta q(Y)} }{p^*(X) c^*(X;\eta) } L(Y, h(X^*, \widehat \beta)) \Big|X \right]\right] \\
&-  \frac{1}{\Pr[S=0]} \E\left[\E\left[ \frac{I(S=1) (1 - p^*(X)) e^{\eta q(Y)} }{p^*(X) c^*(X;\eta) } \dfrac{\E \big[ L(Y, h(X^*, \widehat \beta))e^{\eta q(Y)} | X, S = 1 \big]}{\E \big[ e^{\eta q(Y)} | X, S = 1 \big]} \Big|X \right]\right] \\
&= \frac{1}{\Pr[S=0]} \E\left[  \frac{\Pr[S=1|X] (1 - p^*(X))}{p^*(X) c^*(X;\eta) } \E\left[e^{\eta q(Y)} L(Y, h(X^*, \widehat \beta)) |X,S=1 \right]\right] \\
&- \frac{1}{\Pr[S=0]} \E\left[ \frac{\Pr[S=1|X] (1 - p^*(X)) }{p^*(X) c^*(X;\eta)}  \dfrac{\E \big[ L(Y, h(X^*, \widehat \beta))e^{\eta q(Y)} | X, S = 1 \big]}{\E \big[ e^{\eta q(Y)} | X, S = 1 \big]} \E[e^{\eta q(Y)}  \Big|X,S=1] \right] \\
&=0.
\end{align*}
Combing the above with 
\begin{equation*}
\frac{1}{\Pr[S=0]} \E\left[ I(S=0) b^*(X;\eta)\right] = \E\left[ \dfrac{\E \big[ L(Y, h(X^*, \widehat \beta))e^{\eta q(Y)} | X, S = 1 \big]}{\E \big[ e^{\eta q(Y)} | X, S = 1 \big]}\Big|S=0\right] = \phi(\eta)
\end{equation*}
gives 
\[
\widehat \phi_{aug}(\eta) \overset{P}{\longrightarrow} \phi(\eta).
\]
\textbf{Case 2:} Now assume that $p^*(X) = \Pr[S=1|X]$ and $c^*(X;\eta) = \E [ e^{\eta q(Y)} | X, S = 1]$, but we make no assumptions on the limit $b^*(X;\eta)$. We have
\begin{align*}
&\frac{1}{\Pr[S=0]} \E\left[ \frac{I(S=1) (1 - p^*(X)) e^{\eta q(Y)} }{p^*(X) c^*(X;\eta) } L(Y, h(X^*, \widehat \beta))\right] \\
&= \frac{1}{\Pr[S=0]} \E\left[ \frac{I(S=1) \Pr[S=0|X]}{\Pr[S=1|X] \E[e^{\eta q(Y)}|X, S=1]} e^{\eta q(Y)} L(Y, h(X^*, \widehat \beta))\right] \\
&= \frac{1}{\Pr[S=0]} \E\left[\E\left[ \frac{I(S=1) \Pr[S=0|X]}{\Pr[S=1|X] \E[e^{\eta q(Y)}|X, S=1]} e^{\eta q(Y)} L(Y, h(X^*, \widehat \beta))\Bigg|X \right]\right]\\
&= \frac{1}{\Pr[S=0]} \E\left[\E\left[ \frac{ \Pr[S=0|X]}{\E[e^{\eta q(Y)}|X, S=1]} e^{\eta q(Y)} L(Y, h(X^*, \widehat \beta))\Bigg|X,S=1 \right]\right]\\
&= \frac{1}{\Pr[S=0]} \E\left[ \frac{ \Pr[S=0|X]}{\E[e^{\eta q(Y)}|X, S=1]} \E[e^{\eta q(Y)} L(Y, h(X^*, \widehat \beta))|X,S=1 ]\right] \\
&= \frac{1}{\Pr[S=0]} \E\left[  I(S=0) \frac{\E[e^{\eta q(Y)} L(Y, h(X^*, \widehat \beta))|X,S=1 ]}{\E[e^{\eta q(Y)}|X, S=1]}\right]\\
&= \E\left[  \frac{\E[e^{\eta q(Y)} L(Y, h(X^*, \widehat \beta))|X,S=1 ]}{\E[e^{\eta q(Y)}|X, S=1]}\Bigg|S=0\right] \\
&= \phi(\eta).
\end{align*}
Similar arguments give
\begin{align*}
&\frac{1}{\Pr[S=0]} \E\left[\frac{I(S=1) (1 - p^*(X)) e^{\eta q(Y)} }{p^*(X) c^*(X;\eta) } b^*(X;\eta)\right] \\
&= \frac{1}{\Pr[S=0]} \E\left[ \frac{I(S=1) \Pr[S=0|X]}{\Pr[S=1|X] \E[e^{\eta q(Y)}|X, S=1]} e^{\eta q(Y)} b^*(X;\eta)\right]\\
&= \frac{1}{\Pr[S=0]} \E\left[\E\left[ \frac{I(S=1) \Pr[S=0|X]}{\Pr[S=1|X] \E[e^{\eta q(Y)}|X, S=1]} e^{\eta q(Y)} b^*(X;\eta)\Bigg|X\right]\right]\\
&= \frac{1}{\Pr[S=0]} \E\left[ \frac{Pr[S=0|X]}{\E[e^{\eta q(Y)}|X, S=1]} \E[e^{\eta q(Y)}|X, S=1] b^*(X;\eta)\right]\\
&= \frac{1}{\Pr[S=0]} \E\left[ Pr[S=0|X] b^*(X;\eta)\right] \\
&= \frac{1}{\Pr[S=0]} \E\left[ I(S=0) b^*(X;\eta)\right].
\end{align*}
Combing this gives
\[
\widehat \phi_{aug}(\eta) \overset{P}{\longrightarrow} \phi(\eta).
\]
Combing case 1 and case 2 gives that if either $\widehat b(X; \eta)$ is consistent or both $\widehat p(X)$ and $\widehat c(X;\eta)$ are consistent, then the augmented estimator in non-nested designs is consistent.
\\
\\
\noindent \textbf{Asymptotic representation:}
Rewrite
\begin{align*}
\sqrt{n} (\widehat \phi_{aug}(\eta) - \phi(\eta)) &= \Big(\mathbb{G}_n(H(X,S,Y; \widehat b(X; \eta),\widehat c(X;\eta), \widehat p(X),  n/n_0)) \\
&- \mathbb{G}_n(H(X,S,Y; b^*(X;\eta), c^*(X;\eta), p^*(X), \Pr[S=0]^{-1}) \Big) \\
&+ \mathbb{G}_n(H(X,S,Y; b^*(X;\eta), c^*(X;\eta), p^*(X), \Pr[S=0])^{-1}) \\
& + \sqrt{n} \left[\E[H(X,S,Y; \widehat b(X; \eta), \widehat c(X;\eta), \widehat p(X), n/n_0)] - \phi(\eta)\right].
\end{align*}
By the Donsker assumption
\[
\mathbb{G}_n(H(X,S,Y; \widehat b(X; \eta), \widehat c(X;\eta), \widehat p(X),  n/n_0))
- \mathbb{G}_n(H(X,S,Y; b^*(X;\eta), c^*(X;\eta), p^*(X), \Pr[S=0]^{-1}) = o_P(1).
\]
and by the central limit theorem 
\[
\mathbb{G}_n(H(X,S,Y; b^*(X;\eta), c^*(X;\eta), p^*(X), \Pr[S=0]^{-1}))
\]
is asymptotically normal. So, the behavior of the estimator depends on the 
\[
\sqrt{n} \left(\E[H(X,S,Y; \widehat b(X; \eta), \widehat c(X;\eta), \widehat p(X), n/n_0)] - \phi(\eta)\right).
\]
We have 
\begin{align*}
\sqrt{n} (\E[H(X,S,Y; &\widehat b(X; \eta), \widehat c(X;\eta), \widehat p(X), n/n_0)] - \phi(\eta)) \\ &=     \sqrt{n} (\E[H(X,S,Y; \widehat b(X; \eta), \widehat c(X;\eta), \widehat p(X), \Pr[S=0]^{-1}) - \phi(\eta)) + O_P(1)
\end{align*}
Rewrite $\phi(\eta)$ as
\begin{align}
\label{psi-re}
\phi(\eta) &= \frac{1}{\Pr[S=0]}  \E\left[I(S=0)\dfrac{\E \big[ L(Y, h(X^*, \widehat \beta))e^{\eta q(Y)} | X, S = 1 \big]}{\E \big[ e^{\eta q(Y)} | X, S = 1 \big]}\right] \nonumber 
\\ &= \frac{1}{\Pr[S=0]}\E\left[\Pr[S=0|X] \dfrac{\E \big[ L(Y, h(X^*, \widehat \beta))e^{\eta q(Y)} | X, S = 1 \big]}{\E \big[ e^{\eta q(Y)} | X, S = 1 \big]}\right] 
\end{align}
and 
\begin{equation}
\label{T1}
\frac{1}{\Pr[S=0]} \E[I(S=0) \widehat b(X; \eta)] = \frac{1}{\Pr[S=0]} \E[\Pr[S=0|X] \widehat b(X; \eta)]
\end{equation}
and 
\begin{align}
\label{T2}
&\frac{1}{\Pr[S=0]}\E\left[\dfrac{I(S = 1) (1 - \widehat p(X)) e^{\eta q(Y)} }{\widehat p(X) \widehat c(X;\eta) } \times \left( L(Y, h(X^*, \widehat \beta))-  \widehat b(X; \eta)\right)\right] \nonumber \\
&=\frac{1}{\Pr[S=0]}\E\left[\E\left[\dfrac{I(S = 1) (1 - \widehat p(X)) e^{\eta q(Y)} }{\widehat p(X) \widehat c(X;\eta) } \times \left( L(Y, h(X^*, \widehat \beta))-  \widehat b(X; \eta)\right)\Bigg|X\right] \right] \nonumber \\
&=\frac{1}{\Pr[S=0]}\E\left[\E\left[\dfrac{\Pr[S=1|X] (1 - \widehat p(X)) e^{\eta q(Y)} }{\widehat p(X) \widehat c(X;\eta) } \times \left( L(Y, h(X^*, \widehat \beta))-  \widehat b(X; \eta)\right)\Bigg|X, S=1\right] \right] \nonumber \\
&=\frac{1}{\Pr[S=0]}\E\left[\dfrac{\Pr[S=1|X] (1 - \widehat p(X)) \E[e^{\eta q(Y)}|X,S=1]}{\widehat p(X) \widehat c(X;\eta) } \times \left( \frac{\E[L(Y, h(X^*, \widehat \beta))e^{\eta q(Y)}|X,S=1]}{\E[e^{\eta q(Y)}|X,S=1]}-  \widehat b(X; \eta)\right)\right] 
\end{align}
Combining expressions \eqref{psi-re}, \eqref{T1}, \eqref{T2} gives
\begin{align*}
&\sqrt{n} (\E[H(X,S,Y; \widehat b(X; \eta), \widehat c(X;\eta), \widehat p(X), \Pr[S=0]^{-1}) - \phi(\eta)) \\ &= \frac{1}{\Pr[S=0]}\E\Bigg[ \left( \frac{\E[L(Y, h(X^*, \widehat \beta))e^{\eta q(Y)}|X,S=1]}{\E[e^{\eta q(Y)}|X,S=1]}-  \widehat b(X; \eta)\right) \\
& \quad \quad \times \left( \dfrac{\Pr[S=1|X] (1 - \widehat p(X)) \E[e^{\eta q(Y)}|X,S=1]}{\widehat p(X) \widehat c(X;\eta) } - \Pr[S=0|X]\right)\Bigg]
\end{align*}
Using the assumptions made there exists a constant $M$ such that 
\begin{align*}
& \dfrac{\Pr[S=1|X] (1 - \widehat p(X)) \E[e^{\eta q(Y)}|X,S=1]}{\widehat p(X) \widehat c(X;\eta) } - \Pr[S=0|X] \\
&\leq M \left( \Pr[S=1|X] (1-\widehat p(X)) \E[e^{\eta q(Y)}|X,S=1]  - \Pr[S=0|X] \widehat p(X) \widehat c(X;\eta)\right).
\end{align*}
Little algebra and the Cauchy–Schwarz inequality give
\begin{align*}
&\sqrt{n} (\widehat \phi_{aug}(\eta) - \phi(\eta)) \\ &\leq O_P\Bigg(1 +  \Bigg|\Bigg|\frac{\E[L(Y, h(X^*, \widehat \beta))e^{\eta q(Y)}|X,S=1]}{\E[e^{\eta q(Y)}|X,S=1]}-  \widehat b(X; \eta)\Bigg|\Bigg|_2 \times ||\Pr[S=1|X] - \widehat p(X)||_2 \\
&+ \Bigg|\Bigg|\frac{\E[L(Y, h(X^*, \widehat \beta))e^{\eta q(Y)}|X,S=1]}{\E[e^{\eta q(Y)}|X,S=1]}-  \widehat b(X; \eta)\Bigg|\Bigg|_2 \times ||\E[e^{\eta q(Y)}|X,S=1]-  \widehat c(X;\eta)||_2\Bigg).
\end{align*}
\end{proof}

\subsection{Proof of Theorem \ref{non-nest-dr-2}:}

\begin{proof}
Let $g^*(X)$ and $p^*(X)$ be the asymptotic limits of $\widehat g(X)$ and $\widehat p(X)$, respectively, and define
\[
b^*(X;\eta) = \left(  \dfrac{  L(1, h(X^*, \widehat \beta)) e^{\eta} g^*(X) + L(0, h(X^*, \widehat \beta))  ( 1 - g^*(X))  }{  e^{\eta} g^*(X) + ( 1 - g^*(X))  } \right).
\]
We have that
\[ 
\widehat \phi_{aug}(\eta) \overset{P}{\longrightarrow} \frac{\E\left[ I(S=0) b^*(X;\eta) -  \frac{I(S=1) (1 - p^*(X)) e^{\eta Y} }{p^*(X) (e^{\eta } g^*(X) + ( 1 - g^*(X)) )} \left(L(Y, h(X^*, \widehat \beta)) - b^*(X;\eta) \right)\right]}{\Pr[S=0]}.
\]

Assume that $g^*(X) = \Pr[Y=1|X, S=1]$, but we do not assume that $p^*(X) = \Pr[S=1|X]$. As $g^*(X) =  \Pr[Y=1|X, S=1]$, we have that $b^*(X;\eta) = \dfrac{\E[L(Y, h(X^*, \widehat \beta))e^{\eta Y} | X, S = 1]}{\E[e^{\eta Y} | X, S = 1]}$ and it follows that
\[
\frac{1}{\Pr[S=0]} \E[I(S=0) b^*(X;\eta)] = \E[b^*(X;\eta) |S=0] = \phi(\eta)
\]
Properties of conditional expectations give
\begin{align*}
& \E\left[\frac{I(S=1) (1 - p^*(X)) e^{\eta Y} }{p^*(X) (e^{\eta } g^*(X) + ( 1 - g^*(X)) )} L(Y, h(X^*, \widehat \beta)) \right] \\
&= \E\left[\E\left[\frac{I(S=1) (1 - p^*(X)) e^{\eta Y} }{p^*(X) (e^{\eta } g^*(X) + ( 1 - g^*(X)) )} L(Y, h(X^*, \widehat \beta)) \Bigg| X \right]\right] \\
&= \E\left[\frac{\Pr[S=1|X] (1 - p^*(X))  }{p^*(X) (e^{\eta } g^*(X) + ( 1 - g^*(X)) )} \E[e^{\eta Y} L(Y, h(X^*, \widehat \beta))| X,S=1] \right]
\end{align*}
and
\begin{align*}
& \E\left[\frac{I(S=1) (1 - p^*(X)) e^{\eta Y} }{p^*(X) (e^{\eta } g^*(X) + ( 1 - g^*(X)) )} b^*(X;\eta) \right] \\
&= \E\left[ \E\left[\frac{I(S=1) (1 - p^*(X)) e^{\eta Y} }{p^*(X) (e^{\eta } g^*(X) + ( 1 - g^*(X)))} \dfrac{\E[L(Y, h(X^*, \widehat \beta))e^{\eta Y} | X, S = 1]}{\E[e^{\eta Y} | X, S = 1]} \Bigg|X\right]\right] \\
&= \E\left[ \E\left[\frac{\Pr[S=1|X] (1 - p^*(X)) e^{\eta Y} }{p^*(X) (e^{\eta } g^*(X) + ( 1 - g^*(X)))} \dfrac{\E[L(Y, h(X^*, \widehat \beta))e^{\eta Y} | X, S = 1]}{\E[e^{\eta Y} | X, S = 1]} \Bigg|X,S=1\right]\right] \\
&= \E\left[ \frac{\Pr[S=1|X] (1 - p^*(X)) \E[e^{\eta Y}|X,S=1] }{p^*(X) (e^{\eta } g^*(X) + ( 1 - g^*(X)))} \dfrac{\E[L(Y, h(X^*, \widehat \beta))e^{\eta Y} | X, S = 1]}{\E[e^{\eta Y} | X, S = 1]}\right] \\
&= \E\left[ \frac{\Pr[S=1|X] (1 - p^*(X)) }{p^*(X) (e^{\eta } g^*(X) + ( 1 - g^*(X)))} \E[L(Y, h(X^*, \widehat \beta))e^{\eta Y} | X, S = 1]\right] \\
\end{align*}
Combing the above gives
\[
\widehat \phi_{aug}(\eta) \overset{P}{\longrightarrow} \phi(\eta)
\]
\end{proof}

\subsection{Proof of Theorem \ref{non-nest-dr-1}:}

We use $\overset{P}{\longrightarrow}$ to denote convergence in probability.

\begin{proof}
 Let $a^*(X)$ and $b^*(X;\eta)$ be the asymptotic limits of $\widehat a(X;\eta)$ and $\widehat b(X; \eta)$, respectively.
We have that 
\[
\widehat \phi_{aug}(\eta) \overset{P}{\longrightarrow} \frac{1}{\Pr[S=0]} \E\left[ I(S=0) b^*(X;\eta) +  I(S=1) e^{a^*(X) + \eta q(Y)} \left(L(Y, h(X^*, \widehat \beta)) - b^*(X;\eta) \right)\right]
\]
Now we show that the right hand side of the above equation is equal to $\phi(\eta)$ if at least one of $\widehat b(X; \eta)$ or $\widehat a(X;\eta)$ is consistent.
\\
\textbf{Case 1:} Assume that 
\[
b^*(X;\eta) = \dfrac{\E \big[ L(Y, h(X^*, \widehat \beta))e^{\eta q(Y)} | X, S = 1 \big]}{\E \big[ e^{\eta q(Y)} | X, S = 1 \big]}
\]
but we do not assume that the estimator $\widehat a(X;\eta)$ is consistent. Then we have
\begin{align*}
\frac{1}{\Pr[S=0]} &\E\left[ I(S=1) e^{a^*(X) + \eta q(Y)} \left(L(Y, h(X^*, \widehat \beta)) - b^*(X;\eta)\right)\right] \\
&= \frac{1}{\Pr[S=0]} \E\left[\E\left[ I(S=1) e^{a^*(X) + \eta q(Y)} \left(L(Y, h(X^*, \widehat \beta)) - b^*(X;\eta)\right)\Big|X \right]\right] \\
&= \frac{1}{\Pr[S=0]} \E\left[\E\left[ I(S=1) e^{a^*(X) + \eta q(Y)} L(Y, h(X^*, \widehat \beta)) \Big|X \right]\right] \\
&-  \frac{1}{\Pr[S=0]} \E\left[\E\left[ I(S=1) e^{a^*(X) + \eta q(Y)} \dfrac{\E \big[ L(Y, h(X^*, \widehat \beta))e^{\eta q(Y)} | X, S = 1 \big]}{\E \big[ e^{\eta q(Y)} | X, S = 1 \big]} \Big|X \right]\right] \\
&= \frac{1}{\Pr[S=0]} \E\left[ \Pr[S=1|X] e^{a^*(X)} \E\left[e^{\eta q(Y)} L(Y, h(X^*, \widehat \beta)) |X,S=1 \right]\right] \\
&- \frac{1}{\Pr[S=0]} \E\left[ \Pr[S=1|X] e^{a^*(X)} \dfrac{\E \big[ L(Y, h(X^*, \widehat \beta))e^{\eta q(Y)} | X, S = 1 \big]}{\E \big[ e^{\eta q(Y)} | X, S = 1 \big]} \E[e^{\eta q(Y)}  \Big|X,S=1] \right] \\
&=0.
\end{align*}
Combing the above with 
\begin{equation*}
\frac{1}{\Pr[S=0]} \E\left[ I(S=0) b^*(X;\eta)\right] = \E\left[ \dfrac{\E \big[ L(Y, h(X^*, \widehat \beta))e^{\eta q(Y)} | X, S = 1 \big]}{\E \big[ e^{\eta q(Y)} | X, S = 1 \big]}\Big|S=0\right] = \phi(\eta)
\end{equation*}
gives 
\[
\widehat \phi_{aug}(\eta) \overset{P}{\longrightarrow} \phi(\eta).
\]
\textbf{Case 2:} Now assume that
\[
a^*(X) = \mbox{logit} \big(\Pr[S = 0 | X]\big) -  \ln \E [ e^{\eta q(Y)} | X, S = 1 ],
\]
but we don't make any assumptions on the limit $b^*(X;\eta)$. First
\begin{align*}
&\frac{1}{\Pr[S=0]} \E\left[ I(S=1) e^{a^*(X) + \eta q(Y)} L(Y, h(X^*, \widehat \beta))\right] \\
&= \frac{1}{\Pr[S=0]} \E\left[ \frac{I(S=1) \Pr[S=0|X]}{\Pr[S=1|X] \E[e^{\eta q(Y)}|X, S=1]} e^{\eta q(Y)} L(Y, h(X^*, \widehat \beta))\right] \\
&= \frac{1}{\Pr[S=0]} \E\left[\E\left[ \frac{I(S=1) \Pr[S=0|X]}{\Pr[S=1|X] \E[e^{\eta q(Y)}|X, S=1]} e^{\eta q(Y)} L(Y, h(X^*, \widehat \beta))\Bigg|X \right]\right]\\
&= \frac{1}{\Pr[S=0]} \E\left[\E\left[ \frac{ \Pr[S=0|X]}{\E[e^{\eta q(Y)}|X, S=1]} e^{\eta q(Y)} L(Y, h(X^*, \widehat \beta))\Bigg|X,S=1 \right]\right]\\
&= \frac{1}{\Pr[S=0]} \E\left[ \frac{ \Pr[S=0|X]}{\E[e^{\eta q(Y)}|X, S=1]} \E[e^{\eta q(Y)} L(Y, h(X^*, \widehat \beta))|X,S=1 ]\right] \\
&= \frac{1}{\Pr[S=0]} \E\left[  I(S=0) \frac{\E[e^{\eta q(Y)} L(Y, h(X^*, \widehat \beta))|X,S=1 ]}{\E[e^{\eta q(Y)}|X, S=1]}\right]\\
&= \E\left[  \frac{\E[e^{\eta q(Y)} L(Y, h(X^*, \widehat \beta))|X,S=1 ]}{\E[e^{\eta q(Y)}|X, S=1]}\Bigg|S=0\right] \\
&= \phi(\eta).
\end{align*}
Similar arguments give
\begin{align*}
&\frac{1}{\Pr[S=0]} \E\left[ I(S=1) e^{a^*(X) + \eta q(Y)} b^*(X;\eta)\right] \\
&= \frac{1}{\Pr[S=0]} \E\left[ \frac{I(S=1) \Pr[S=0|X]}{\Pr[S=1|X] \E[e^{\eta q(Y)}|X, S=1]} e^{\eta q(Y)} b^*(X;\eta)\right]\\
&= \frac{1}{\Pr[S=0]} \E\left[\E\left[ \frac{I(S=1) \Pr[S=0|X]}{\Pr[S=1|X] \E[e^{\eta q(Y)}|X, S=1]} e^{\eta q(Y)} b^*(X;\eta)\Bigg|X\right]\right]\\
&= \frac{1}{\Pr[S=0]} \E\left[ \frac{Pr[S=0|X]}{\E[e^{\eta q(Y)}|X, S=1]} \E[e^{\eta q(Y)}|X, S=1] b^*(X;\eta)\right]\\
&= \frac{1}{\Pr[S=0]} \E\left[ Pr[S=0|X] b^*(X;\eta)\right] \\
&= \frac{1}{\Pr[S=0]} \E\left[ I(S=0) b^*(X;\eta)\right].
\end{align*}
Combing this gives
\[
\widehat \phi_{aug}(\eta) \overset{P}{\longrightarrow} \phi(\eta).
\]
Combing case 1 and case 2 gives that if at least one of $\widehat b(X; \eta)$ or $\widehat a(X;\eta)$ are consistent, then the augmented estimator in non-nested designs is consistent.
\end{proof}

\section{Asymptotic properties for the augmented estimators under a nested design}
\label{app-nest}

The augmented estimator for a nested design is given by
\begin{equation*}
  \begin{split}
\widehat \psi_{aug}(\eta) &= \frac{1}{n} \sum_{i=1}^n \left( S_i L(Y_i, h(X_i^*, \widehat \beta))  + I(S_i=0) \widehat b(X_i; \eta) + \dfrac{I(S_i = 1) (1 - \widehat p(X_i)) e^{\eta q(Y_i)}}{\widehat p(X_i) \widehat c(X_i;\eta)}\left(L(Y_i, h(X_i^*, \widehat \beta)) - \widehat b(X_i; \eta) \right) \right).
	\end{split}
\end{equation*}
For arbitrary functions $ b'(X), c'(X)$, and $p'(X)$ define
\begin{align*}
&G(X,S,Y; b'(X), c'(X), p'(X)) \\ &=S L(Y, h(X^*, \widehat \beta))  + I(S=0) b'(X) + \dfrac{I(S = 1) (1 - p'(X)) e^{\eta q(Y)}}{ p'(X) c'(X)}\left(L(Y, h(X^*, \widehat \beta)) - b'(X) \right).
\end{align*}
Using this notation we can write the augmented estimator as
\[
\widehat \psi_{aug}(\eta) = \mathbb{P}_n(G(X,S,Y; \widehat b(X; \eta), \widehat c(X;\eta), \widehat p(X))).
\]
For nested designs we make the following assumptions
\begin{enumerate}
\item[C1] At least one of 
\[
\widehat b(X; \eta) \overset{P}{\longrightarrow}   \dfrac{\E \big[ L(Y, h(X^*, \widehat \beta))e^{\eta q(Y)} | X, S = 1 \big]}{\E \big[ e^{\eta q(Y)} | X, S = 1 \big]}
\]
or
\[
(\widehat p(X), \widehat c(X;\eta)) \overset{P}{\longrightarrow}  (\Pr[S=1|X], \E[e^{\eta q(Y)}|X, S=1]).
\]
\item[C2] The sequences $G(X,S,Y; b^*(X;\eta), c^*(X;\eta), p^*(X))$ and $G(X,S,Y; \widehat b(X; \eta), \widehat c(X;\eta), \widehat p(X))$ are Donsker.
\item[C3] \[
||G(X,S,Y; \widehat b(X; \eta), \widehat c(X;\eta), \widehat p(X)) - G(X,S,Y; b^*(X;\eta), c^*(X;\eta), p^*(X))||_2 \overset{P}{\longrightarrow} 0.
\]
\item[C4] $\E[G(X,S,Y; b^*(X;\eta), c^*(X;\eta), p^*(X))^2] < \infty$.
\item[C5] The functions $\widehat p(X)$ and $\widehat c(X;\eta)$ are uniformly bounded away from zero and $\E[L(Y, h(X^*, \widehat \beta))e^{\eta q(Y)}|X,S=1]$, $\E[e^{\eta q(Y)}|X,S=1]$, $\widehat b(X; \eta)$ are uniformly bounded.
\end{enumerate}
\begin{theorem}
\label{thm-as-rep-nes}
If assumptions C1-C5 hold, then the augmented estimator for nested designs:
\begin{itemize}
    \item Is consistent. That is, $\widehat \psi_{aug}(\eta) \overset{P}{\longrightarrow} \psi(\eta)$.
    \item Has the asymptotic representation
    \begin{equation}
    \label{thm-1-as}
    \sqrt{n} (\widehat \psi_{aug}(\eta) - \psi(\eta)) = \mathbb{G}_n(G(X,S,Y; b^*(X;\eta), c^*(X;\eta), p^*(X)) + Rem,
    \end{equation}
where the reminder term satisfies
\begin{align*}
Rem  \leq O_P\Bigg(1 &+ \sqrt{n}  \Bigg|\Bigg|\frac{\E[L(Y, h(X^*, \widehat \beta))e^{\eta q(Y)}|X,S=1]}{\E[e^{\eta q(Y)}|X,S=1]}-  \widehat b(X; \eta)\Bigg|\Bigg|_2 \times ||\Pr[S=1|X] - \widehat p(X)||_2 \\
&+ \sqrt{n} \Bigg|\Bigg|\frac{\E[L(Y, h(X^*, \widehat \beta))e^{\eta q(Y)}|X,S=1]}{\E[e^{\eta q(Y)}|X,S=1]}-  \widehat b(X; \eta)\Bigg|\Bigg|_2 \times ||\E[e^{\eta q(Y)}|X,S=1]-  \widehat c(X;\eta)||_2\Bigg).
\end{align*}
\end{itemize}
\end{theorem}
The asymptotic representation of the estimator has the same implication for rate of convergence of the estimator as for the non-nested designs. Now we prove Theorem \ref{thm-as-rep-nes}.
\begin{proof}
\textbf{Consistency:} 
We have that 
\begin{align*}
\widehat \psi_{aug}(\eta) &\overset{P}{\longrightarrow} \E[I(S=1) L(Y, h(X^*, \widehat \beta))] \\ &+ \E\left[ I(S=0) b^*(X;\eta) +  \frac{I(S=1) (1 - p^*(X)) e^{\eta q(Y)} }{p^*(X) c^*(X;\eta) )} \left(L(Y, h(X^*, \widehat \beta)) - b^*(X;\eta) \right)\right].
\end{align*}
Now we show that the right hand side of the above equation is equal to $\psi(\eta)$ if either $\widehat b(X; \eta)$ is consistent or both $\widehat p(X)$ and $\widehat c(X;\eta)$ are consistent.
\\
\\
\textbf{Case 1:} First assume that 
\[
\widehat b(X; \eta) \overset{P}{\longrightarrow}   \dfrac{\E \big[ L(Y, h(X^*, \widehat \beta))e^{\eta q(Y)} | X, S = 1 \big]}{\E \big[ e^{\eta q(Y)} | X, S = 1 \big]},
\]
but we make no assumptions on the limits $p^*(X)$ and $c^*(X;\eta)$. As shown for the non-nested case
\[
\E\left[\frac{I(S=1) (1 - p^*(X)) e^{\eta q(Y)} }{p^*(X) c^*(X;\eta) } \left(L(Y, h(X^*, \widehat \beta)) - b^*(X;\eta)\right)\right] = 0
\]
Hence,
\begin{align*}
\widehat \psi_{aug}(\eta) &\overset{P}{\longrightarrow} \E[I(S=1) L(Y, h(X^*, \widehat \beta))] \\ &+ \E\left[ I(S=0) \dfrac{\E \big[ L(Y, h(X^*, \widehat \beta))e^{\eta q(Y)} | X, S = 1 \big]}{\E \big[ e^{\eta q(Y)} | X, S = 1 \big]} \right] = \psi(\eta).
\end{align*}
\textbf{Case 2:} Now assume that $p^*(X) = \Pr[S=1|X]$ and $c^*(X;\eta) = \E [ e^{\eta q(Y)} | X, S = 1]$, but we make no assumptions on the limit $b^*(X;\eta)$. Similar to the non-nested design, we have
\begin{align*}
&\E\left[ \frac{I(S=1) (1 - p^*(X)) e^{\eta q(Y)} }{p^*(X) c^*(X;\eta) } L(Y, h(X^*, \widehat \beta))\right] \\
&= \E\left[  I(S=0) \frac{\E[e^{\eta q(Y)} L(Y, h(X^*, \widehat \beta))|X,S=1 ]}{\E[e^{\eta q(Y)}|X, S=1]}\right]
\end{align*}
And
\[
\E\left[\frac{I(S=1) (1 - p^*(X)) e^{\eta q(Y)} }{p^*(X) c^*(X;\eta) } b^*(X;\eta)\right] = \E\left[ I(S=0) b^*(X;\eta)\right].
\]
Combing this gives
\[
\widehat \psi_{aug}(\eta) \overset{P}{\longrightarrow} \psi(\eta).
\]
Combing case 1 and case 2 gives that if either $\widehat b(X; \eta)$ is consistent or both $\widehat p(X)$ and $\widehat c(X;\eta)$ are consistent, then the augmented estimator in nested designs is consistent.
\\
\\
\noindent \textbf{Asymptotic representation:} Rewrite
\begin{align*}
\sqrt{n} (\widehat \psi_{aug}(eta) - \psi(\eta)) &= \Big(\mathbb{G}_n(G(X,S,Y; \widehat b(X; \eta), \widehat c(X;\eta), \widehat p(X))) \\
&- \mathbb{G}_n(G(X,S,Y; b^*(X;\eta), c^*(X;\eta), p^*(X))) )\Big) \\
&+ \mathbb{G}_n(G(X,S,Y; b^*(X;\eta), c^*(X;\eta), p^*(X))) \\
& + \sqrt{n} \left[\E[G(X,S,Y; \widehat b(X; \eta), \widehat c(X;\eta), \widehat p(X)] - \psi\right].
\end{align*}
As for the non-nested design the behavior of the estimator depends 
\[
\sqrt{n} \left(\E[G(X,S,Y; \widehat b(X; \eta), \widehat c(X;\eta), \widehat p(X))] - \psi(\eta)\right).
\]
Rewrite $\psi(\eta)$ as
\begin{equation*}
\psi(\eta) = \E[I(S=1) L(Y, h(X^*, \widehat \beta))] + \E\left[\Pr[S=0|X] \dfrac{\E \big[ L(Y, h(X^*, \widehat \beta))e^{\eta q(Y)} | X, S = 1 \big]}{\E \big[ e^{\eta q(Y)} | X, S = 1 \big]}\right] 
\end{equation*}
and the asymptotic representation follows using the same arguments as for the non-nested case.
\end{proof}
As for non-nested designs we can use the alternative parameterization of the augmented estimator to rewrite it as
\begin{equation*}
  \begin{split}
\widehat \psi_{aug}(\eta) &= \frac{1}{n} \sum_{i=1}^n \left( S_i L(Y_i, h(X_i^*, \widehat \beta))  + I(S_i=0) \widehat b(X_i; \eta) + I(S_i=1) e^{a(X_i;\eta) + \eta q(Y_i)} \left(L(Y_i, h(X_i^*, \widehat \beta)) - \widehat b(X_i; \eta) \right) \right).
	\end{split}
\end{equation*}
The following theorem shows the doubly robustness of the augmented estimator under a nested design using the alternative parameterization
\begin{theorem}
If at least one of $\widehat b(X; \eta)$ or $\widehat a(X;\eta)$ are consistent, then the augmented estimator for a nested designs is consistent (i.e.,~$\widehat \psi_{aug}(\eta) \overset{P}{\longrightarrow} \psi(\eta)$).
\end{theorem}
\begin{proof}
We have 
\begin{equation*}
  \begin{split}
\widehat \psi_{aug}(\eta) \overset{P}{\longrightarrow} \E\left[ S L(Y, h(X^*, \widehat \beta))  + I(S=0) b^*(X;\eta) + I(S=1) e^{a^*(X) + \eta q(Y)} \left(L(Y, h(X^*, \widehat \beta)) - b^*(X;\eta) \right) \right].
	\end{split}
\end{equation*}
\textbf{Case 1:} Assume that 
\[
b^*(X;\eta) = \dfrac{\E \big[ L(Y, h(X^*, \widehat \beta))e^{\eta q(Y)} | X, S = 1 \big]}{\E \big[ e^{\eta q(Y)} | X, S = 1 \big]},
\]
but we do not assume that the estimator $\widehat a(X;\eta)$ is consistent. In the proof of Theorem \ref{non-nest-dr-1} we showed that
\[
\E\left[ I(S=1) e^{a^*(X) + \eta q(Y)} \left(L(Y, h(X^*, \widehat \beta)) - b^*(X;\eta)\right)\right] =0.
\]
And by expression \eqref{eq_mean_identification}
\begin{align*}
\psi(\eta) = \E\left[ S L(Y, h(X^*, \widehat \beta))  + I(S=0) \dfrac{\E \big[ L(Y, h(X^*, \widehat \beta))e^{\eta q(Y)} | X, S = 1 \big]}{\E \big[ e^{\eta q(Y)} | X, S = 1 \big]}\right].
\end{align*}
Completing the proof for Case 1.

\textbf{Case 2:} Now assume that
\[
a^*(X) = \mbox{logit} \big(\Pr[S = 0 | X]\big) -  \ln \E [ e^{\eta q(Y)} | X, S = 1 ],
\]
but we don't make any assumptions on the limit $b^*(X;\eta)$. As in the proof of Theorem \ref{non-nest-dr-1}
\[
\E\left[ I(S=1) e^{a^*(X) + \eta q(Y)} L(Y, h(X^*, \widehat \beta))\right] = \E\left[  I(S=0) \frac{\E[e^{\eta q(Y)} L(Y, h(X^*, \widehat \beta))|X,S=1 ]}{\E[e^{\eta q(Y)}|X, S=1]}\right]
\]
and 
\[
 \E\left[ I(S=1) e^{a^*(X) + \eta q(Y)} b^*(X;\eta)\right] = \E\left[ Pr[S=0|X] b^*(X;\eta)\right] = \E\left[ I(S=0) b^*(X;\eta)\right].
\]
Combining this gives 
\[
\widehat \psi_{aug}(\eta) \overset{P}{\longrightarrow} \E\left[ S L(Y, h(X^*, \widehat \beta))\right] + \E\left[  I(S=0) \frac{\E[e^{\eta q(Y)} L(Y, h(X^*, \widehat \beta))|X,S=1 ]}{\E[e^{\eta q(Y)}|X, S=1]}\right] = \psi(\eta).
\]
From Case 1 and Case 2 above we have that if at least one of $\widehat b(X; \eta)$ or $\widehat a(X;\eta)$ is consistent, then $\widehat \psi_{aug}(\eta) \overset{P}{\longrightarrow}  \psi(\eta)$.
\end{proof}

\section{Stability analysis using the Coronary Artery Study data}
\label{CASS:SWA}

In this section we present results of two stability analysis for the CASS data. First, we used generalized additive models are used to estimate  $\Pr[Y = 1 | X, S = 1]$ and $\Pr[S=1|X]$ and second we used the Jackknife to construct confidence intervals. 

Figures \ref{CassGlmJK} shows results when the Jackknife is used for estimating 95\% confidence intervals and logistic regression was used to estimate $\Pr[Y = 1 | X, S = 1]$ and $\Pr[S=1|X]$. Figure \ref{GamBS} shows results when generalized additive models are used to estimate $\Pr[Y = 1 | X, S = 1]$ and $\Pr[S=1|X]$ and the non-parametric bootstrap is used for estimating 95\% confidence intervals. Figure \ref{GamJK} shows results when generalized additive models are used to estimate $\Pr[Y = 1 | X, S = 1]$ and $\Pr[S=1|X]$ and the Jackknife is used for estimating 95\% confidence intervals. The results show that confidence intervals based on the non-parametric bootstrap are slightly wider than confidence intervals based on the Jackknife and that the values of both point estimates and confidence intervals were not sensitive to whether  logistic regression or generalized additive models were used to estimate $\Pr[Y = 1 | X, S = 1]$ and $\Pr[S=1|X]$.

\begin{figure}[ht]
    \centering
    \includegraphics[width=\textwidth]{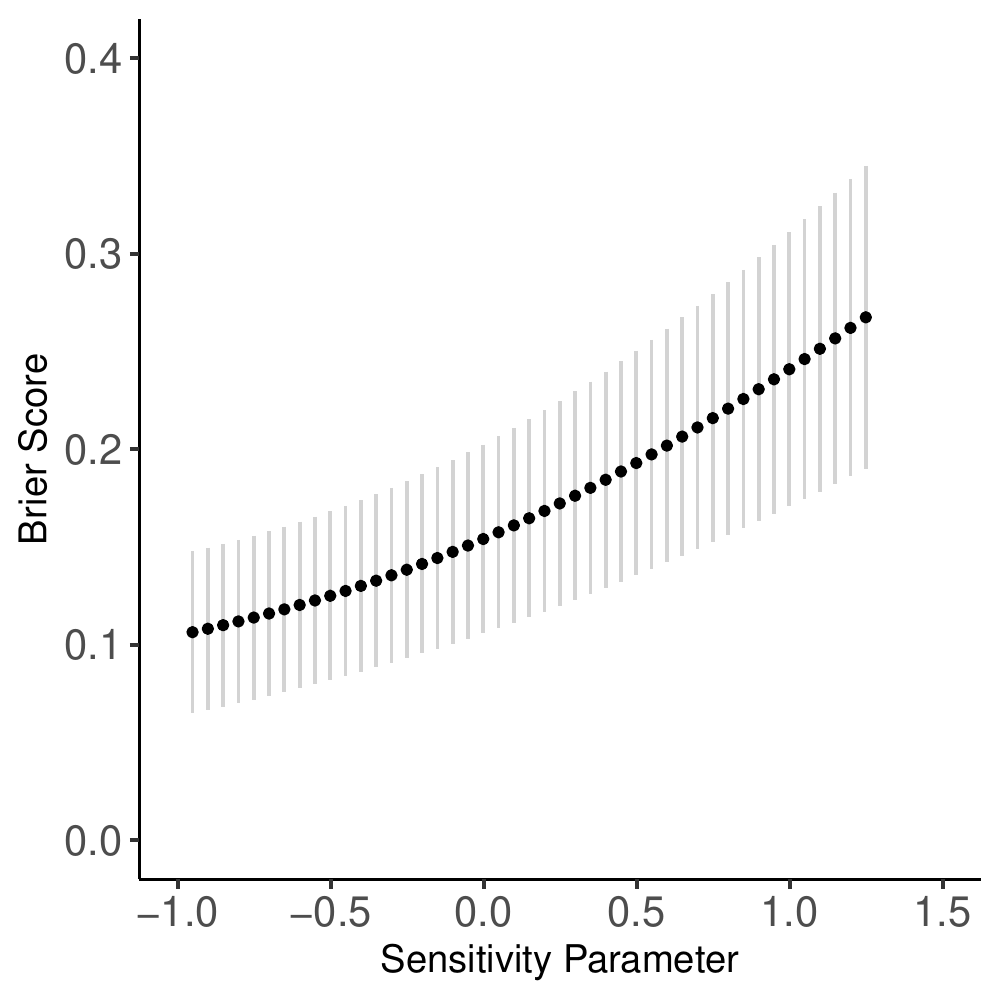}
    \caption{Sensitivity analysis using CASS data. The value used for the sensitivity parameter ($\eta$) is on the x-axis and the corresponding doubly robust estimates of the Brier risk are on the y-axis. The nuisance functions $\Pr[Y = 1 | X, S = 1]$ and $\Pr[S=1|X]$ were estimated using logistic regression models and the 95\% confidence intervals were calculated using the Jackknife. The solid line connects point estimates and the gray lines are point-wise 95\% confidence intervals.}
    \label{CassGlmJK}
\end{figure}

\begin{figure}[ht]
    \centering
    \includegraphics[width=\textwidth]{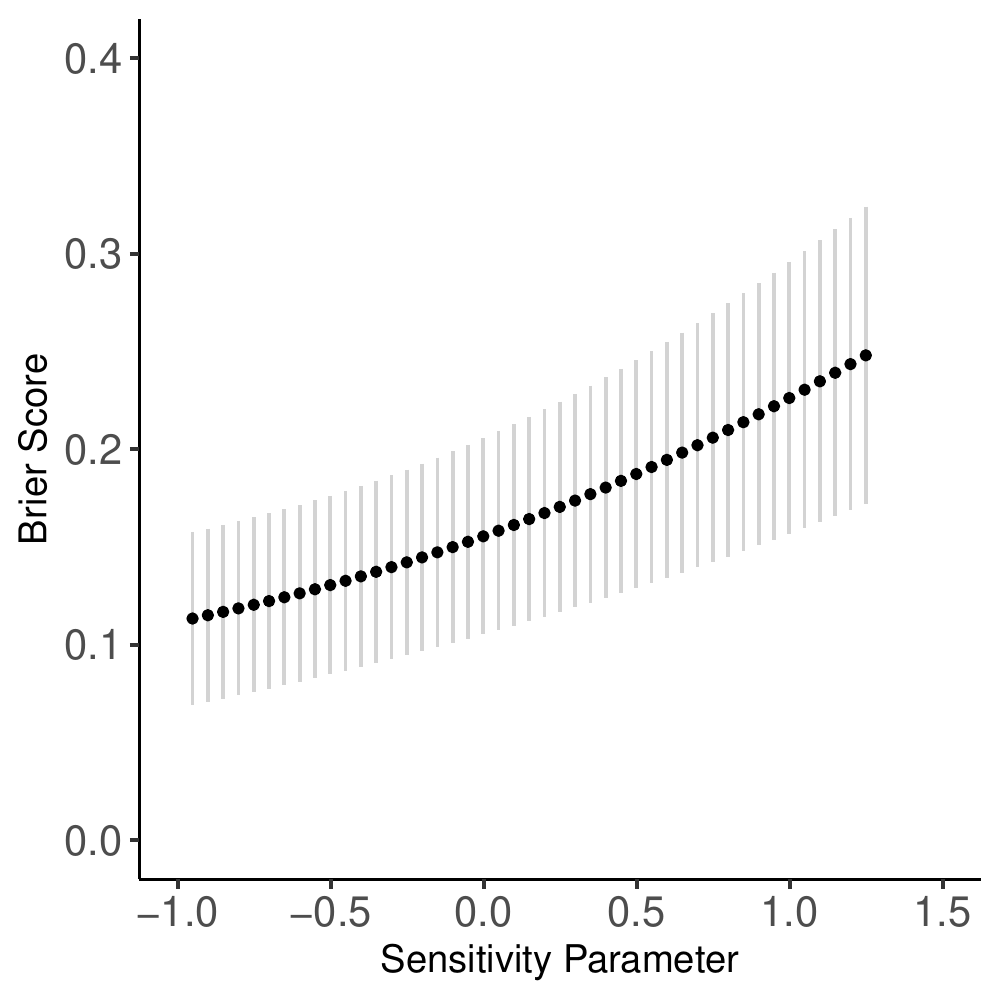}
    \caption{Sensitivity analysis using CASS data. The value used for the sensitivity parameter ($\eta$) is on the x-axis and the corresponding doubly robust estimates of the Brier risk are on the y-axis. The nuisance functions $\Pr[Y = 1 | X, S = 1]$ and $\Pr[S=1|X]$ were estimated using generalized additive models and the 95\% confidence intervals were calculated using the non-parametric bootstrap. The solid line connects point estimates and the gray lines are point-wise 95\% confidence intervals.}
    \label{GamBS}
\end{figure}

\begin{figure}[ht]
    \centering
    \includegraphics[width=\textwidth]{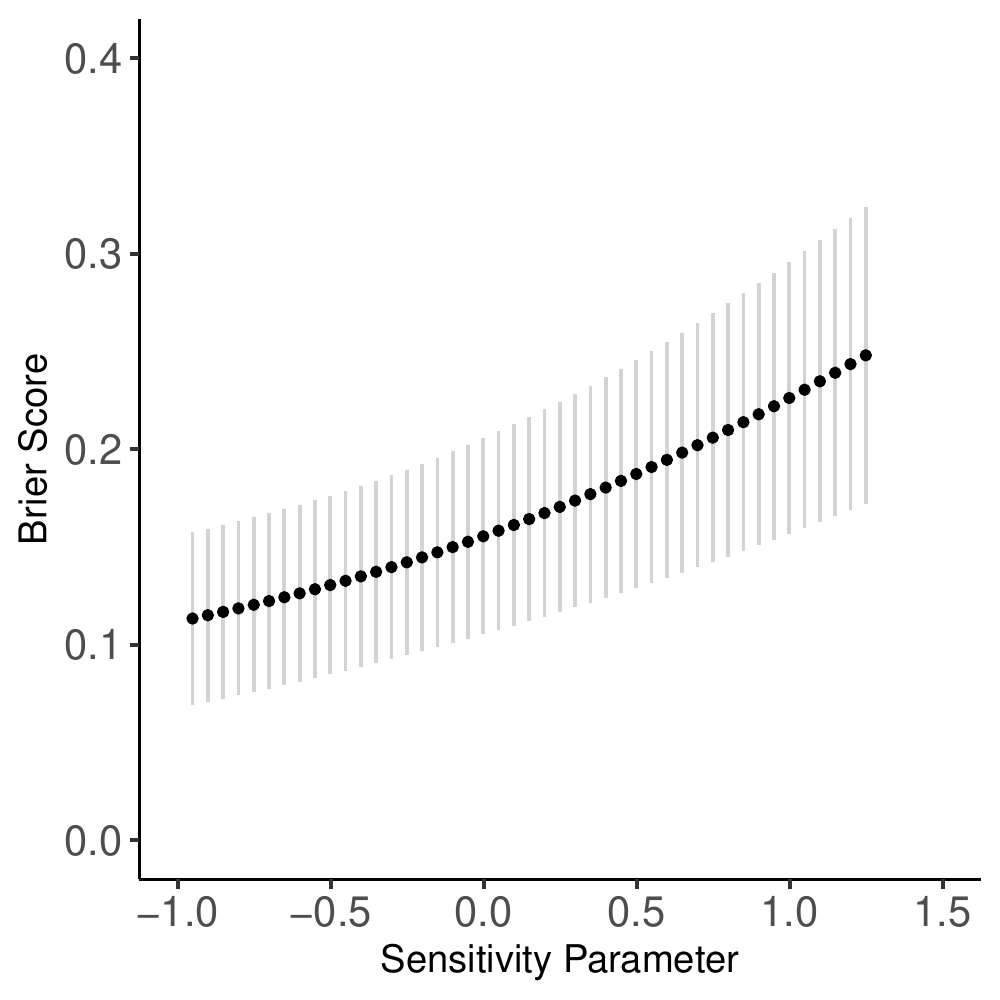}
    \caption{Sensitivity analysis using CASS data. The value used for the sensitivity parameter ($\eta$) is on the x-axis and the corresponding doubly robust estimates of the Brier risk are on the y-axis. The nuisance functions $\Pr[Y = 1 | X, S = 1]$ and $\Pr[S=1|X]$ were estimated using generalized additive models and the 95\% confidence intervals were calculated using the Jackknife. The solid line connects point estimates and the gray lines are point-wise 95\% confidence intervals.}
    \label{GamJK}
\end{figure}

\end{document}